\documentclass[preprint,12pt,authoryear]{elsarticle}

\usepackage{multibib}
\newcites{Supp}{REFERENCES}
\usepackage{amsfonts}
\usepackage{amsmath}
\usepackage{amssymb}
\usepackage{amsthm}
\usepackage{mathtools}
\usepackage{graphicx}
\usepackage{xcolor}
\usepackage[authoryear]{natbib}
\usepackage{algorithm}
\usepackage{subcaption}
\usepackage{booktabs}
\usepackage[noend]{algpseudocode}
\usepackage{caption}
\usepackage[toc,page,title]{appendix}

\newcommand{\showrevisions}{0}
\newif\ifshowrevisions
\showrevisionstrue   

\newcommand{\newtext}[1]{%
    \if1\showrevisions%
        \textcolor{purple}{#1}%
    \else%
        #1%
    \fi%
}
\theoremstyle{plain}
\newtheorem{proposition}{Proposition}

\theoremstyle{remark}
\newtheorem{remark}{Remark}

\begin{document}

\begin{frontmatter}

\title{Neural Orientation Distribution Fields for Estimation and Uncertainty Quantification in Diffusion MRI}

\author[label1]{William Consagra, Lipeng Ning, Yogesh Rathi}

\affiliation[label1]{organization={Psychiatry Neuroimaging Laboratory, Brigham and Women’s Hospital, Harvard Medical School},
                 addressline={399 Revolution Drive},
                 city={Boston},
                 postcode={02215},
                 state={MA},
             country={United States}}

\begin{abstract}
Inferring brain connectivity and structure \textit{in-vivo} requires accurate estimation of the orientation distribution function (ODF), which encodes key local tissue properties. However, estimating the ODF from diffusion MRI (dMRI) signals is a challenging inverse problem due to obstacles such as significant noise, high-dimensional parameter spaces, and sparse angular measurements. In this paper, we address these challenges by proposing a novel deep-learning based methodology for continuous estimation and uncertainty quantification of the spatially varying ODF field. We use a neural field (NF) to parameterize a random series representation of the latent ODFs, implicitly modeling the often ignored but valuable spatial correlation structures in the data, and thereby improving efficiency in sparse and noisy regimes. An analytic approximation to the posterior predictive distribution is derived which can be used to quantify the uncertainty in the ODF estimate at any spatial location, avoiding the need for expensive resampling-based approaches that are typically employed for this purpose. We present empirical evaluations on both synthetic and real in-vivo diffusion data, demonstrating the advantages of our method over existing approaches. 
\end{abstract}


\end{frontmatter}

\section{Introduction}
The structure of the human brain plays a fundamental role in determining various cognitive phenotypes as well as neurodegenerative and psychiatric disorders. 
To better understand these relationships and facilitate the development of effective treatments and therapies, it is essential to have the ability to map the brain's structure \textit{in vivo}. Diffusion Magnetic Resonance Imaging (dMRI)  is a widely used imaging technique for this purpose \newtext{\citep{basser1996}}. By capturing signals related to the localized diffusion of water molecules, dMRI allows us to investigate tissue composition. Consequently, dMRI signals can provide valuable insights into numerous microstructural properties of interest.
\par 
A primary object of interest in diffusion MRI is the \textit{orientation distribution function} (ODF), which is an antipodally symmetric smooth density function on $\mathbb{S}^2$ characterizing the orientational distribution of water molecule diffusion. Due to the dependence of the angular diffusion of water molecules on the local tissue environment, important biological microstructural properties of interest can be inferred from the ODF. For example, in white matter regions, the modes of the ODF are often used as a surrogate for the dominant directions of coherently aligned collections of myelinated neural axons, referred to as white matter fiber tracts, passing through the region \newtext{\citep{descoteaux2015}}. This local directional information is pieced together in a process known as tractography in order to infer the large scale neural pathways that connect different regions in the brain \newtext{\citep{basser2000,fanzhang2022}}. More generally, the ODF can be linked to various biophysical models of diffusion and used to estimate their parameters, which can then be used as biomarkers in downstream prediction or hypothesis testing procedures \citep{novikov2018,novikov2019quantifying,veraart2020noninvasive}. 
\par 
Estimating the ODF from dMRI signals is a challenging inverse problem with many sources of uncertainty. Being a spatially indexed functional parameter, the ODF is a theoretically infinite dimensional unknown that must be estimated from the necessarily finite number of diffusion signals that can be collected during any diffusion experiment. Further complicating matters are the facts that diffusion data is notoriously noisy \citep{henkelman1985} and that, for emerging  applications such as high (spatial) resolution imaging \citep{wang2021}, acquiring the signals can be expensive and thus the data may exhibit sparsity in the angular domain. Hence, it is crucial to accurately quantify the uncertainty in the ODF estimates, particularly in scenarios characterized by high noise levels or limited angular samples. Furthermore, it is essential to efficiently  propagate this uncertainty to downstream analyses of interest to ensure reliable and robust interpretation of the results \newtext{\citep{siddiqui2021}}. 
\par 
Most of the commonly used approaches estimate the ODFs at all voxels independently \citep{descoteaux2007,michailovich2010,ning2015}, thereby ignoring valuable spatial correlations which can be leveraged to improve statistical efficiency, particularly in the sparse and low SNR regimes. A smaller collection of methods do integrate neighborhood information into estimation. Such methods can largely be grouped into three classes, the first using some variant of kernel smoothing with a data-adaptive bandwidth to identify the local neighborhoods over which to pool information \citep{becker2012,baba2013,becker2014,cabeen2016,chen2019,ye2016,yu2013}, the second utilize energy formulations and promote smoothness with some spatial regularizer \citep{michailovich2011,raj2011,liu2013}, while the third, and arguably most popular approach invokes a two-stage procedure in which the raw diffusion data is first smoothed over local patches and estimation is subsequently performed using the denoised data \citep{veraart2016,grande2019,ramos2021}. 
\par 
To quantify uncertainty in the estimates, bootstrap resampling techniques are the dominant approach \newtext{\citep{jones2008wildboot,berman2008,haroon2008,yap2014}}, though some Bayesian methods have recently been proposed \citep{jens2018}. Resampling approaches can be problematic for sparse samples, as the bootstrapped samples tend to underestimate the true randomness of the distribution \citep{kauermann2009}. Furthermore, the number of voxels in modern imaging applications can be on the order of millions, hence computation and storage can become a bottleneck. Finally, while it is possible to integrate spatial information, e.g. \cite{yap2014}, this serves to further exacerbate computational issues. We also note that such approaches are all heavily dependent on the chosen spatial discretization, hence the method's computational performance may further degrade in the continuous space limit. Many important applications require ``going off the grid'', that is, estimation of the ODF at some unobserved spatial location \newtext{\citep{tournier2012mrtrix}}. This is typically accomplished by repeated application of a post-hoc local interpolation \citep{goh2011}. Even less attention has been paid to adequately quantifying the uncertainty of the interpolates. 
\par 
\subsection{Our Contributions}
In this work, we propose a novel fully-continuous methodology for estimating the spatially varying ODF field. A key aspect of our approach is the implicit modeling of spatial correlations of the latent ODF field by parameterizing its random series decomposition using a deep neural network. Specifically, we use a \textit{neural field} (NF) architecture (also referred to as an implicit neural representation (INR)) \newtext{\citep{sitzmann2020,tancik2020,mancini2022,molaei2023implicit}} to form a continuous parameterization of the ODF field and propose an estimation procedure that remains robust in sparse sample and high noise cases.  
In addition, we propose a novel method for fast and lightweight uncertainty quantification (UQ) for our deep ODF estimator. 
Drawing inspiration from stochastic last-layer approaches \citep{snoek2015,matthews2017}, we employ a Gaussian process assumption and derive the closed-form conditional predictive distribution for the ODF at any spatial location. To maintain computational efficiency, we formulate fast point estimators for the unknown conditioning parameters and plug them into the derived predictive distribution. 
\par 
The remainder of the paper is organized as follows. Section~\ref{sec:overview} gives a brief overview of diffusion magnetic resonance imaging and the ODF inverse problem. Section~\ref{sec:models} develops our modeling framework for the latent ODF field, including a statistical model for the observed data and parametric approximations using a NF model. Section~\ref{sec:stat_inference} outlines our proposed statistical inference procedure and discusses how to quantify relevant uncertainties. Implementation details and empirical evaluation, including simulation studies for both 2D and 3D phantoms and real in-vivo diffusion data, are reported in Section~\ref{sec:experiments}. \newtext{Section~\ref{sec:discussion} provides additional discussions and directions for future research, and Section~\ref{sec:conclusion} concludes.}

\section{Diffusion Magnetic Resonance Imaging}\label{sec:overview}

We begin with a discussion outlining relevant background on diffusion MRI. Denote the imaging domain $\Omega\subset\mathbb{R}^{D}$, $D\in\{2,3\}$. The ensemble average diffusion propagator (EAP) at $\boldsymbol{v}\in\Omega$, denoted $E_{\boldsymbol{v}}$, is a probability density function describing the distribution of the average displacement of proton spins during a given experimental diffusion time. The ODF at $\boldsymbol{v}$, denoted $g_{\boldsymbol{v}}$, is explicitly defined as the radial integration of the EAP, given by:
\begin{equation}\label{eqn:EAP_ODF}
    g_{\boldsymbol{v}}(\boldsymbol{p}) = \int_{0}^\infty E_{\boldsymbol{v}}(r\boldsymbol{p})r^2dr \quad \boldsymbol{p}\in\mathbb{S}^2.
\end{equation}
The MRI signals can be sensitized to these local diffusions using appropriate acquisition sequences. In particular, the \textit{diffusion signal attenuation}, referred to from here on as the diffusion signal and denoted as $f_{\boldsymbol{v}}(\boldsymbol{p}, b)$, collected using Pulsed Gradient Spin-Echo sequences, shares a Fourier relationship with the EAP: 
\begin{equation}\label{eqn:EAP}
    f_{\boldsymbol{v}}(\boldsymbol{p}, b) = \int_{\mathbb{R}^{3}} E_{\boldsymbol{v}}(\boldsymbol{r})\text{exp}(2\pi i b\boldsymbol{p}^{\intercal}\boldsymbol{r})d\boldsymbol{r},
\end{equation}
under the narrow pulse approximation on the magnetic gradients \citep{ombao2016}. Notice that the diffusion signal is parameterized by two experimentally controlled acquisition parameters: the direction of the applied the magnetic field gradient, which can be represented as a point on $\boldsymbol{p}\in\mathbb{S}^2$, and a non-negative composite scalar parameter $b\in\mathbb{R}^{+}$, referred to as the $b$-value, which is related to the strength, duration and timing of the applied magnetic gradients, along with the proton gyromagnetic ratio.
Equations~\eqref{eqn:EAP} and \eqref{eqn:EAP_ODF} suggest a seemingly natural approach to ODF estimation via a two stage procedure: first estimating the EAP from the discrete Fourier transform applied to a lattice sampling of diffusion signals and then computing the numerical integration in the radial direction. This is referred to as diffusion spectrum imaging (DSI). Unfortunately, DSI requires dense Cartesian sampling over a 3D grid of $b \; \text{and} \; \boldsymbol{p}$, resulting in long scanning times.
\par
Alternatively, for a fixed non-zero $b^{*}$, \cite{tuch2004} show that the ODF is approximately equal to the Funk-Radon transform of the diffusion signal over a fixed radial $b^{*}$-shell:
\begin{equation}\label{eqn:FRT}
    g_{\boldsymbol{v}}(\boldsymbol{p}) \approx \int_{\mathbb{S}^2}\delta(\boldsymbol{p}^{\intercal}\boldsymbol{u})f_{\boldsymbol{v}}(\boldsymbol{u},b^{*})d\boldsymbol{u}.
\end{equation}
Using the Funk-Hecke theorem, it can be shown that the \textit{real-symmetric spherical harmonic functions}, denoted here as $\{\phi_0, \phi_1, ..., \}$, are non-zero eigenfunctions of the Funk-Radon transform. These functions are defined according to:
\begin{equation}\label{eqn:real_sph_harm}
    \phi_j = 
      \begin{cases}
      \sqrt{2}\text{Re}(Y_k^m) & -k\le m < 0\\
      Y_k^0 & m=0\\
      \sqrt{2}\text{Img}(Y_k^m) & 0 < m \le k
    \end{cases} \quad k = 0,2,4,...,l;
\end{equation}
where 
$$
    Y_l^m(\alpha_1, \alpha_2) = \sqrt{\frac{(2l+1)(l-m)!}{4\pi (l+m)!}}P_l^m(cos(\alpha_1))e^{im\alpha_2} \ ,
$$
are the spherical harmonics, which form a complete orthogonal basis system for $L^2(\mathbb{S}^2)$,
under spherical parameterization with polar angle $\alpha_1\in [0, \pi]$ and azimuthal angle $\alpha_2\in [0, 2\pi]$. $P_m^l$ are the Legendre polynomials with order indices $l = 0, 1, ..., $ and phase factors $m = -l, ..., 0, ..., l$ and the indices of the $\{\phi_j\}$ are defined by $j = (k^2+k+2)/2 + m-1$. The associated eigenvalue of $\phi_j$ is given by $2\pi P_{l_{j}}(0)$, with $l_j$ indicating the order of the harmonic $\phi_j$. The $\{\phi_j\}$ then form a basis for the set symmetric spherical functions:
$$
    \mathcal{H}:=\{h\in L^2(\mathbb{S}^2): h(\boldsymbol{p}) = h(-\boldsymbol{p})\}.
$$ 
Since the diffusion signals are antipodally symmetric and the Funk-Radon transform annihilates all odd functions, it follows that $g_{\boldsymbol{v}},f_{\boldsymbol{v}}\in\mathcal{H}$ and hence both can be represented using an expansion over the real-symmetric spherical harmonic basis. If we further assume the elements in $\mathcal{H}$ to be continuous, the Funk-Radon transform has an inverse \citep{quellmalz2020}, denoted here as $\mathcal{G}$, and it follows that the signal is related to the ODF via 
$f_{\boldsymbol{v}} := \mathcal{G}\left[g_{\boldsymbol{v}}\right]$. Coupling the aforementioned spectral properties of the Funk-Radon transform with the fact that any element in $\mathcal{H}$ can be represented using an expansion over the real-symmetric spherical harmonic basis, it follows that representing $g_{\boldsymbol{v}}$ and $f_{\boldsymbol{v}}$ using some finite expansion of the $\phi_j$'s allows for the derivation of a simple diagonal linear map between the two functions (see Section~\ref{apx:theory} of the supplement). This approach avoids the costly lattice sampling designs and potentially unstable numerical integration and has thereby become the preferred approach in practice for single shell data \citep{descoteaux2007}.

\section{A Model for the Continuous Orientation Density Field}\label{sec:models}
In this section, we outline our modeling framework. We begin in Section~\ref{ssec:latent_density_model} by formulating a probabilistic model for the latent ODF field. In Section~\ref{ssec:param_model}, we propose a continuous field parameterization using a deep neural network and in Section~\ref{ssec:statistical_model} we outline the statistical model for the observed data. 

\subsection{Model for the Latent Field}\label{ssec:latent_density_model}
Let $(\Sigma, \mathcal{B}(\Sigma), \mathbb{P})$ be a probability space. For a given spatial location $\boldsymbol{v}\in\Omega$, we model the latent ODF as \textit{random function} $g_{\boldsymbol{v}}:\Sigma\mapsto\mathcal{H}$, i.e. a random variable with realizations in $\mathcal{H}$ under the standard $L^2(\mathbb{S}^2)$ inner product, denoted as $\langle, \rangle_{\mathcal{H}}$, which is a Hilbert space. Due to the underlying biophysics, a complex spatial covariance structure exists \textit{between} the functions observed at different locations $\boldsymbol{v}$. Such correlation can be accounted for by modeling the process as a function-valued random field indexed by $\boldsymbol{v}$: $\{g_{\boldsymbol{v}}: \boldsymbol{v}\in\Omega\}$
that is, a random field such that each $g_{\boldsymbol{v}}:=g(\boldsymbol{v},\cdot):\Sigma\mapsto \mathcal{H}$ is a random function \citep{menafoglio2013,martinezhernadez2020}. 
\par 
The mean and covariance function of the process $g_{\boldsymbol{v}}$ are defined as:
\begin{equation}\label{eqn:angular_random_function}
\begin{aligned}
    \mu({\boldsymbol{v}},\boldsymbol{p)} &= \mathbb{E}[g_{\boldsymbol{v}}(\boldsymbol{p})] := \int_{\Sigma} g_{\boldsymbol{v}}(\sigma, \boldsymbol{p})\mathbb{P}(d\sigma) \\
    C_{\mathcal{H}}(\boldsymbol{v},\boldsymbol{p}_1,\boldsymbol{p}_2) &= \mathbb{E}[(g(\boldsymbol{v}, \boldsymbol{p}_1) - \mu(\boldsymbol{v},\boldsymbol{p}_1))(g(\boldsymbol{v},\boldsymbol{p}_2) - \mu(\boldsymbol{v},\boldsymbol{p}_2))] \quad\boldsymbol{p}_1,\boldsymbol{p}_2\in \mathbb{S}^2,
    \end{aligned}
\end{equation}
where the moments can be rigorously defined using the Bochner integral \citep{hsing2015}. We encode our \emph{a-priori} assumptions on the angular properties of $g_{\boldsymbol{v}}$ by specifying a form for the parameters $\mu$ and $C_{\mathcal{H}}$ under a Gaussian process model \citep{rasmussen2005}. We adopt a constant mean model: $\mu(\boldsymbol{v},\boldsymbol{p}) := \mu(\boldsymbol{v})$, which is interpreted as an isotropic diffusion field with scale $\mu(\boldsymbol{v})$. To model the second order behavior, we assume a stationary and isotropic (rotationally invariant) prior angular covariance model, implying that the covariance function $C_{\mathcal{H}}$ along any two directions is a function purely of the angle between them, which is a natural prior model in the absence of additional spatial context \citep{andersson2015}. The following proposition characterizes a class of such stationary and isotropic covariance functions which respect the antipodal symmetry of $\mathcal{H}$.
\begin{proposition}\label{prop:eigen_analysis_S2}
    Assume the covariance function $C_{\mathcal{H}}$ is rotationally invariant and that realizations of $g \sim g\in\mathcal{H}$ w.p.1. Then the correlation function is given by
    \begin{equation}\label{eqn:zonal_mercer_kernel}
        \text{Cor}(g_{\boldsymbol{v}}(\boldsymbol{p}_1),g_{\boldsymbol{v}}(\boldsymbol{p}_2)) = \sum_{k=1}^\infty \text{s}_{\gamma}(\sqrt{l_k(l_{k}+1)})\phi_k(\boldsymbol{p}_1)\phi_k(\boldsymbol{p}_2)
    \end{equation}
    where $\text{s}_{\gamma}$ is the spectral density function  of the kernel of $C_{\mathcal{H}}$ with parameters $\gamma$.
\end{proposition}
\noindent{Proposition~\ref{prop:eigen_analysis_S2} is useful from a computational perspective, as it diagonalizes the covariance kernels of interest over the real-symmetric harmonics and facilities a principled manner of constraining the number of parameters determining $C_{\mathcal{H}}$. Specifically, we use the spherical Mat\'{e}rn family \citep{guinness2016} to define the angular prior covariance function, the spectral density of which is given by 
$$
\text{s}_{\gamma}(\omega) = \frac{2^{3}\pi^{3/2}\Gamma(\nu + \frac{3}{2})(2\nu)^\nu}{\Gamma(\nu)\rho^{2\nu}}\left(\frac{2\nu}{\rho^2} + 4\pi^2\omega^2\right)^{-(\nu + \frac{3}{2})}\quad \gamma = (\nu, \rho),
$$
with parameters $\nu$, which controls the smoothness, and $\rho$, dictating the length-scale of correlation, and $\Gamma$ is the gamma function.}
\par 
The global spatial dependence structure between any pair of locations $\boldsymbol{v}_1,\boldsymbol{v}_2\in\Omega$ is determined by the spatial covariance function of the process, defined by:
\begin{equation}\label{eqn:statial_covariance}
    C_{\Omega}(\boldsymbol{v}_1, \boldsymbol{v}_2) = \mathbb{E}[\langle g(\boldsymbol{v}_1,\cdot) - \mu(\boldsymbol{v}_1,\cdot), g(\boldsymbol{v}_2,\cdot) - \mu(\boldsymbol{v}_2,\cdot)\rangle_{\mathcal{H}}].
\end{equation}
As $C_{\Omega}$ is a positive symmetric function over a $2(D)$-dimensional domain, for tractability of parameter estimation, i.e. to avoid the curse of dimensionality, one approach is to invoke prior assumptions of stationarity and isotropy to reduce the dimensionality of the parameter space defining $C_{\Omega}$ and estimate it directly. Such assumptions imply a spatial correlation structure that is constant on spherical contours, which is violated for datasets exhibiting complex anisotropic dependence structures, e.g. of the type that would be encountered along  boundaries between white matter bundles in diffusion MRI. 
\par 
Alternatively, consider the random series decomposition of spatial fields:
\begin{equation}\label{eqn:field_lin_model}
\begin{aligned}
    g_{\boldsymbol{v}}(\boldsymbol{p}) &:= g(\boldsymbol{v},\boldsymbol{p}) = \mu(\boldsymbol{v}) + \sum_{k=1}^\infty c_k(\boldsymbol{v})\phi_k(\boldsymbol{p}) 
    \approx\mu(\boldsymbol{v}) + \sum_{k=1}^K c_k(\boldsymbol{v})\phi_k(\boldsymbol{p}),
\end{aligned}
\end{equation}
where $c_k(\boldsymbol{v})$ are random functions with realizations in a suitable function space over $\Omega$. The rank $K$ approximation assumes that the contribution of all high-frequency harmonic basis functions above a certain degree are negligible. The mean function $\mu$ can be considered as the coefficient field for the harmonic $\phi_0$, the constant function on $\mathbb{S}^2$. Intuitively, model \eqref{eqn:field_lin_model} can be interpreted as a spatial linear mixed-effects model,
which decomposes the field into an angularly constant mean function $\mu(\boldsymbol{v})$ (isotropic part)
and K spatially dependent random coefficient field deviation functions $\boldsymbol{c}(\boldsymbol{v})$ (anisotropic part).
\par 
The vector-valued random function $\boldsymbol{c}(\boldsymbol{v}) := (c_1(\boldsymbol{v}),...,c_{K}(\boldsymbol{v}))^\intercal$ determines the spatial distribution of anisotropy in the diffusion field. Using the orthonormality of the harmonics, it is easy to show that the spatial covariance is implicitly defined by the moments of the coefficients of the basis expansion \eqref{eqn:field_lin_model}:
$C_{\Omega}(\boldsymbol{v}_1, \boldsymbol{v}_2) = \sum_{k=1}^K\mathbb{E}[c_k(\boldsymbol{v}_1)c_k(\boldsymbol{v}_2)]$. 
Hence, if we are able to approximate the distribution of the coefficient fields, we may avoid either direct flexible modeling of $C_{\Omega}$ (and the computational issues that come along with it) or the imposition of unrealistic modeling assumptions, and instead translate the burden of flexible process modeling to the inductive bias of our model for the coefficient field. In the following Section, we propose a flexible model for the coefficient fields using deep-basis functions. 

\subsection{A Parameterization using Deep Basis Function}\label{ssec:param_model}
We are interested in modeling the random coefficient field functions $c_k(\boldsymbol{v})$. As the $c_{k}$'s are infinite dimensional random variables, we require some discretization for tractable computation. For this purpose, we propose to use another layer of basis expansion, now in the spatial domain. Specifically, let $\boldsymbol{\xi}_{\boldsymbol{\theta}}:\Omega \mapsto \mathbb{R}^{r}$ be a rank $r$-basis system, depending on a parameter $\boldsymbol{\theta}\in\Theta\subset\mathbb{R}^{p}$. We model the spatial coefficient field via multivariate basis expansion
\begin{equation}\label{eqn:last_layer_expansion}
    \boldsymbol{c}(\boldsymbol{v}) := \boldsymbol{W}\boldsymbol{\xi}_{\boldsymbol{\theta}}(\boldsymbol{v}), \text{ where } \boldsymbol{W}\in\mathbb{R}^{K\times r}.
\end{equation}
In contrast to the appealing properties that lead to the choice of the harmonic basis system for the angular model, there is no such immediately obvious choice for $\boldsymbol{\xi}_{\boldsymbol{\theta}}$. The basis should be flexible enough to handle both spatially smooth regions as well as adequately approximate sharp discontinuities corresponding to tissue boundaries, along with being well defined over arbitrary subvolumes of $\Omega$. Often used basis systems for multidimensional function representation such as tensor-product splines and finite element basis functions are problematic for our situation, as the former implies that $\Omega$ naturally decomposes into a product space, while the approximation performance of the latter is heavily dependent on the selection of a triangulation of the domain $\Omega$, which is itself a difficult problem 
\citep{lai2007}.
\par 
Instead, we proposed a data-driven approach in which we parameterize $\boldsymbol{\xi}_{\boldsymbol{\theta}}$ as an $L$-layer neural network, whose parameters $\boldsymbol{\theta}$ are to be estimated from the observed data. Such neural networks that directly parameterize continuous functions have been collectively referred to as neural fields (NF) or implicit neural representations (INR) \citep{xie2022}. The proposed neural network basis system is defined according to a multi-layer perception (MLP) architecture, which has the general form:
\begin{equation}\label{eqn:inr_spatial_field}
    \begin{aligned}
        \boldsymbol{x}^{(0)} &= \alpha(\boldsymbol{v}) \\
        \boldsymbol{x}^{(l)} &= \rho^{(l)}(\boldsymbol{W}^{(l)}\boldsymbol{x}^{(l-1)} + \boldsymbol{b}^{(l)})\quad l=1,...,L \\
        \boldsymbol{\xi}(\boldsymbol{v}) &= \boldsymbol{x}^{(L)} 
    \end{aligned}
\end{equation}
where $\boldsymbol{x}^{(l)}\in\mathbb{R}^{i_{l}}$ is the output from the $l-1$'st layer, $\rho^{(l)}$ is the $l$'th activation function and $\alpha(\boldsymbol{v}) = (\alpha_1(\boldsymbol{v}), ..., \alpha_{d_{0}}(\boldsymbol{v}))$ is an initial $d_0$-dimensional mapping. Note that by construction, $i_{L}=r$. Denote the parameter $\boldsymbol{\theta}:=(\boldsymbol{W}^{(L)}, ..., \boldsymbol{W}^{(1)}, \boldsymbol{b}^{(L)}, ..., \boldsymbol{b}^{(1)})$, i.e. the weights and biases of the network. 
\par  
We take the first layer mapping $\alpha$ to be a sufficiently high-frequency Fourier expansion of the form $\alpha(\boldsymbol{v}) = \text{sin}(\boldsymbol{W}_0\boldsymbol{v} + \boldsymbol{b}_0)$ for some random frequencies 
$\boldsymbol{W}_0\in\mathbb{R}^{d_{0}\times D}$ and phases $\boldsymbol{b}_0\in\mathbb{R}^{d_{0}}$. Such
NFs have been shown empirically to avoid the spectral bias of traditional coordinate-based MLPs \citep{sitzmann2020,tancik2020,mildenhall2021},  i.e. the tendency of learning overly smooth function fits \citep{rahaman2019}. From the perspective of approximation theory, it has recently been shown that NFs with first layer sinusoidal basis functions and continuous activations construct a representation space that is equivalent to a dictionary of sinusoidal functions, where the number of dictionary elements is exponential in the network depth \citep{fathony2021,gizem2022}. This allows relatively shallow architectures to define very high-dimensional function spaces, and thus are an attractive option for flexible modeling of the coefficient field function $\boldsymbol{c}(\boldsymbol{v})$.
\par 
To allow for dependence structure with the fields $\boldsymbol{c}(\boldsymbol{v})$, we model the angularly constant mean field by adding another fully connected output dimension to the network \eqref{eqn:inr_spatial_field}, i.e. $\mu(\boldsymbol{v})=  \boldsymbol{\mu}^{\intercal}\boldsymbol{\xi}_{\boldsymbol{\theta}}(\boldsymbol{v})$, for $\boldsymbol{\mu}\in\mathbb{R}^{r}$. Putting this all together, we have the following parameterized form of \eqref{eqn:field_lin_model}:
\begin{equation}\label{eqn:lin_field_model_params}
    g(\boldsymbol{v},\boldsymbol{p}) = \underbrace{\boldsymbol{\mu}^{\intercal}\boldsymbol{\xi}_{\boldsymbol{\theta}}(\boldsymbol{v})}_{\mu(\boldsymbol{v})} + \underbrace{\boldsymbol{\xi}^{\intercal}_{\boldsymbol{\theta}}(\boldsymbol{v})\boldsymbol{W}^{\intercal}}_{\boldsymbol{c}^{\intercal}(\boldsymbol{v})}\boldsymbol{\phi}(\boldsymbol{p}).
\end{equation} 
\subsection{A Statistical Model for the Observed Data}\label{ssec:statistical_model}
Denote the set $\boldsymbol{V}\subset\Omega$, where $\boldsymbol{v}_1,...,\boldsymbol{v}_{N}\in\boldsymbol{V}$ are the coordinates of $N$ observation locations. For each $i=1,...,N$, we observe noisy signals $\boldsymbol{y}_{i}=(y_{i,1,b_{1}},...,
y_{i,M,b_{M}})^\intercal$ along some common set of angular locations $\boldsymbol{P}_{M}=(\boldsymbol{p}_1,...,\boldsymbol{p}_{M})^\intercal$. From here on we assume that $b_m$ are constant for all $m=1,...,M$ and drop the subscript notation for clarity. From the discussion in Section~\ref{sec:overview}, we know that the true (noiseless) signals are related to the latent orientation density field via the inverse Funk Radon transformation $\mathcal{G}$. We adopt a Gaussian noise model for the measurement error, yielding 
\begin{equation}\label{eqn:stat_model_observed_data}
    \begin{aligned}
    &p(\boldsymbol{y}_1, ..., \boldsymbol{y}_{N}| g, \boldsymbol{V}, \boldsymbol{P}_{M}) = \prod_{i=1}^{N}p(\boldsymbol{y}_{i} | g_{\boldsymbol{v}_{i}},\boldsymbol{P}_{M}), \quad \boldsymbol{y}_{i} | g_{\boldsymbol{v}_{i}},\boldsymbol{P}_{M} \sim  \mathcal{N}\left(\begin{pmatrix}
        \mathcal{G}[g_{\boldsymbol{v}_{i}}](\boldsymbol{p}_1) \\
        \vdots \\
        \mathcal{G}[g_{\boldsymbol{v}_{i}}](\boldsymbol{p}_M)  
    \end{pmatrix}, \sigma_{e}^2\boldsymbol{I}_{M}\right) \\
    \end{aligned}
\end{equation}
for $\sigma_{e}^2>0$ and $\mathcal{N}$ denotes the normal distribution.
\par 
Restricted to the space spanned by the rank $K$ truncation of the harmonics, $\mathcal{G}$ can be represented as a diagonal matrix: $\boldsymbol{G}\in\text{diag}(\mathbb{R}^{K})$ with $k$'th diagonal element $[2\pi P_{l_{k}}(0)]^{-1}$. Denote $\boldsymbol{\Phi}\in\mathbb{R}^{M\times K}$, where $\boldsymbol{\Phi}_{mk} = \phi_k(\boldsymbol{p}_m)$ and define the matrix $\boldsymbol{\Phi}_{G}:=\boldsymbol{\Phi}\boldsymbol{G}$. Coupling the statistical model for the observed data \eqref{eqn:stat_model_observed_data} and the parameterization \eqref{eqn:lin_field_model_params}, we have the per-voxel likelihood 
$$
\boldsymbol{y}_i | \boldsymbol{v}_i, \boldsymbol{W}, \boldsymbol{\theta}, \boldsymbol{\mu}, \sigma_{e}^2 \sim \mathcal{N}(\boldsymbol{\mu}^{\intercal}\boldsymbol{\xi}_{\boldsymbol{\theta}}(\boldsymbol{v}_i)\boldsymbol{1}_{M} + \boldsymbol{\Phi}_{G}\boldsymbol{W}\boldsymbol{\xi}_{\boldsymbol{\theta}}(\boldsymbol{v}_i), \sigma_{e}^2\boldsymbol{I}),
$$
where $\boldsymbol{1}_{M}\in\mathbb{R}^{M}$ is the column-vector of ones. 
Denote $\boldsymbol{Y} = [\boldsymbol{y}_1^{\intercal}, ..., \boldsymbol{y}_{N}^{\intercal}] \in \mathbb{R}^{M\times N}$ and $\boldsymbol{\Xi}_{\boldsymbol{\theta}} = [\boldsymbol{\xi}_{\boldsymbol{\theta}}^{\intercal}(\boldsymbol{v}_1),...,\boldsymbol{\xi}_{\boldsymbol{\theta}}^{\intercal}(\boldsymbol{v}_N)]^{\intercal}\in\mathbb{R}^{r\times N}$, 
the complete-data likelihood can be written 
\begin{equation}\label{eqn:data_likelihood}
\begin{aligned}
    &\boldsymbol{Y}|\boldsymbol{V}, \boldsymbol{W}, \boldsymbol{\theta}, \boldsymbol{\mu}, \sigma_e^2 \sim \mathcal{MN}_{M\times N}(((\boldsymbol{1}_{M}\boldsymbol{\mu}^{\intercal} + \boldsymbol{\Phi}_{G}\boldsymbol{W})\boldsymbol{\Xi}_{\boldsymbol{\theta}}, \sigma_{e}^2\boldsymbol{I}_M, \boldsymbol{I}_N),\\
\end{aligned}
\end{equation}
 where $\mathcal{MN}$ denotes the matrix normal distribution. 
\begin{remark}
It may be of interest to perform estimation and inference on the signal function directly $\mathcal{G}[g]$. This can be easily accommodated in our framework by taking $\boldsymbol{G}=\boldsymbol{I}_{K}$, the identity matrix in $\mathbb{R}^{K}$, and re-interpreting \eqref{eqn:lin_field_model_params} in terms of the diffusion signal function field. 
\end{remark}

\section{Statistical Inference}\label{sec:stat_inference}

Estimation and uncertainty quantification can be accomplished by forming the posterior distribution of the field estimates
$p(g_{\boldsymbol{v}}|\boldsymbol{Y},\boldsymbol{V})$. Unfortunately, full posterior inference is computationally intractable, mainly due to the enormous size of the network parameters $\boldsymbol{\theta}$. To avoid this computational bottleneck, we develop a fast approximate inference procedure, \newtext{wherein we first derive the analytic form of a conditional posterior of interest and subsequently form point-estimates of the remaining unknown conditioning parameters, which are then plugged into the derived analytic form for inference.} 
\subsection{Estimating the Conditional Field Posterior}\label{ssec:conditional_posterior} 
\newtext{To model the uncertainty in the mean predictions, we assume the form of the conditional predictive distribution of the mean function is given by: $\mu(\boldsymbol{v}) |\boldsymbol{\theta}, \boldsymbol{\mu}, \sigma_\mu^2 \sim \mathcal{N}(\boldsymbol{\mu}^{\intercal}\boldsymbol{\xi}_{\boldsymbol{\theta}}(\boldsymbol{v}), \sigma_\mu^2)$, for unknown spatially constant variance $\sigma_\mu^2>0$, and assume conditional independence between the fields: $\mu(\boldsymbol{v})\perp \boldsymbol{c}(\boldsymbol{v}) | \boldsymbol{\theta}$. Under these assumptions and the model formulated in Section~\ref{sec:models}, it can be shown that the predictive posterior of interest has the form:
\begin{equation}\label{eqn:function_space_posterior}
    \begin{aligned}
        g(\boldsymbol{v},\cdot)| \boldsymbol{V}, \boldsymbol{Y}, \boldsymbol{\theta}, \boldsymbol{\mu}, \gamma, \sigma_w^2,  \sigma_e^2, \sigma_\mu^2 \sim \mathcal{GP}\big(&\boldsymbol{\xi}_{\boldsymbol{\theta}}^{\intercal}(\boldsymbol{v})\boldsymbol{\mu} + \boldsymbol{\phi}^{\intercal}(\boldsymbol{p})\frac{1}{\sigma_{e}^2}[\boldsymbol{\xi}_{\boldsymbol{\theta}}^{\intercal}(\boldsymbol{v})\otimes\boldsymbol{I}_{K}]\boldsymbol{\Lambda}_{\boldsymbol{\theta}}^{-1}[\boldsymbol{\Xi}_{\boldsymbol{\theta}}^{\intercal}\otimes\boldsymbol{\Phi}_{G}]^{\intercal}\text{vec}(\boldsymbol{Y}^{(c)}), \\ &\sigma_{\mu}^2 + \boldsymbol{\phi}^{\intercal}(\boldsymbol{p}_1)[\boldsymbol{\xi}_{\boldsymbol{\theta}}^{\intercal}(\boldsymbol{v})\otimes\boldsymbol{I}_{K}]\boldsymbol{\Lambda}_{\boldsymbol{\theta}}^{-1}[\boldsymbol{\xi}_{\boldsymbol{\theta}}^{\intercal}(\boldsymbol{v})\otimes\boldsymbol{I}_{K}]^{\intercal}\boldsymbol{\phi}(\boldsymbol{p}_2)\big),
    \end{aligned}
\end{equation}
where $\boldsymbol{Y}^{(c)}:=\boldsymbol{Y}-\boldsymbol{1}_{M}\boldsymbol{\mu}^{\intercal}\boldsymbol{\Xi}_{\boldsymbol{\theta}}$, $\sigma_w^2>0$ is a prior angular variance parameter, 
$$
\begin{aligned}
        \boldsymbol{\Lambda}_{\boldsymbol{\theta}} &= \frac{1}{\sigma^2_{e}}(\frac{\sigma^2_{e}}{\sigma^2_{w}}\boldsymbol{I}_{r}\otimes\boldsymbol{R}_{\gamma} +  \boldsymbol{\Xi}_{\boldsymbol{\theta}}\boldsymbol{\Xi}_{\boldsymbol{\theta}}^{\intercal}\otimes\boldsymbol{\Phi}_{G}^{\intercal}\boldsymbol{\Phi}_{G}); \text{ } \boldsymbol{R}_{\gamma}^{-1} := \text{Diag}(\text{s}_{\gamma}(\sqrt{l_1(l_{1}+1)}), ...,\text{s}_{\gamma}(\sqrt{l_K(l_{K}+1)})),
\end{aligned}
$$
and $\otimes$ denotes the Kronecker product. Please visit supplemental Section~\ref{apx:theory} for a derivation of this result.} 
\par 
\newtext{Equation~\eqref{eqn:function_space_posterior} is an analytic function of the high-dimensional conditioning parameters $\boldsymbol{\mu},\boldsymbol{\theta}$, along with the variance parameters $\sigma_e^2,\sigma_w^2,\sigma^2_{\mu}$ and angular correlation parameters $\gamma$, which are unknown in practice and must be estimated. Algorithm~\ref{alg:inferece_algo} summarizes our general procedure for obtaining these estimates. For the remainder of this Section, we elaborate on each of the steps.
}
\begin{algorithm}[t]
  \caption{Inference procedure}
  \label{alg:inferece_algo}
  \begin{algorithmic}[1]
    \State Partition data into disjoint sets 
$\{\boldsymbol{Y}, \boldsymbol{V}\} = \{\boldsymbol{Y}_{calib}, \boldsymbol{V}_{calib}\}\bigcup\{\boldsymbol{Y}_{train}, \boldsymbol{V}_{train}\}$
\State Select hyperparameters using Algorithm~\ref{alg:BO_hyperparameter_optimization} (supplemental text) with $\{\boldsymbol{Y}_{train}, \boldsymbol{V}_{train}\}$
\State Form $\widehat{\boldsymbol{\theta}},\widehat{\boldsymbol{\mu}}$ via  \eqref{eqn:penalized_likelihood}, conditioned on selected hyperparameters, with $\{\boldsymbol{Y}_{train}, \boldsymbol{V}_{train}\}$
\State Estimate $\hat{\sigma}_e^2$ using \eqref{eqn:noise_estimator} and $\hat{\sigma}_{w}^2,\hat{\sigma}_{\mu}^2$ via \eqref{eqn:calibration_likelihood_optimization} with $\{\boldsymbol{Y}_{calib}, \boldsymbol{V}_{calib}\}$ 
\State Use point estimate estimates from steps 2-4 and full data $\{\boldsymbol{Y}, \boldsymbol{V}\}$ to condition \eqref{eqn:function_space_posterior}
\end{algorithmic}
\end{algorithm}
\bigskip\par
\newtext{\noindent{\textit{Estimating $\boldsymbol{\mu},\boldsymbol{\theta}$}}:} Estimating the hyperparameters of a Gaussian processes is typically done using some variant of marginal likelihood maximization \citep{rasmussen2005}. Unfortunately, gradient-based optimization of the resulting objective function is problematic for our situation, as it requires taking derivatives of the inverse of an $rK\times rK$ dimensional matrix with respect to the network parameters $\boldsymbol{\theta}$, resulting in a major computational bottleneck for large $r$ and $K$. \newtext{To circumvent this issue, we could instead maximize the likelihood in Equation~\eqref{eqn:data_likelihood} through stochastic gradient descent to obtain $\widehat{\boldsymbol{\mu}}, \widehat{\boldsymbol{W}},\widehat{\boldsymbol{\theta}}$, and then use the point estimates  $\widehat{\boldsymbol{\theta}},\widehat{\boldsymbol{\mu}}$ to condition the posterior of interest \eqref{eqn:function_space_posterior}. 
}
\par 
However, direct maximization of \eqref{eqn:data_likelihood} risks over-fitting the data,  learning basis functions which are too ``wiggly'', particularly in the high noise and/or sparse angular sample cases. To promote an appropriate notion of smoothness in the field, we want to shrink our field estimates toward the angular prior. Under our rank $K$ truncation, we can always define a functional analog of the Mahalanobis distance between any $h\in\mathcal{H}$ and the reduced rank prior via the semi-metric:
$$
\| h \|_{C_{K}}^2 = \frac{1}{\sigma_w^2}\sum_{k=1}^K\frac{\langle h, \phi_k\rangle_{\mathcal{H}}}{\text{s}_{\gamma}(\sqrt{l_k(l_{k}+1)})} \quad h\in\mathcal{H}
$$
\citep{galeano2015}. Notice that we do not penalize the coefficient associated with the constant harmonic $\phi_0$, i.e. the mean field, as not to introduce unnecessary bias, as this field is much easier to estimate than the higher order harmonic fields. Under our parametric model \eqref{eqn:lin_field_model_params}, it is straightforward to show that this penalty has the easily computable form
$$
\left\| \boldsymbol{\mu}^{\intercal}\boldsymbol{\xi}_{\boldsymbol{\theta}}(\boldsymbol{v}) + \boldsymbol{\phi}^{\intercal}(\boldsymbol{p})\boldsymbol{W}\boldsymbol{\xi}_{\boldsymbol{\theta}}(\boldsymbol{v})\right\|_{C_{K}}^2:=\frac{1}{\sigma_w^2}\boldsymbol{\xi}^{\intercal}_{\boldsymbol{\theta}}(\boldsymbol{v})\boldsymbol{W}^{\intercal}\boldsymbol{R}_{\gamma}\boldsymbol{W}\boldsymbol{\xi}_{\boldsymbol{\theta}}(\boldsymbol{v}),
$$
which can be integrated over the domain to form the global penalty for the field estimate. Putting this all together, we aim to maximize the penalized log maximum likelihood
\begin{equation}\label{eqn:penalized_likelihood}
   \max_{\boldsymbol{\mu}, \boldsymbol{W}, \boldsymbol{\theta}} -\left(\text{trace}((\boldsymbol{Y}-\boldsymbol{\Phi}_{G}\boldsymbol{W}\boldsymbol{\Xi}_{\boldsymbol{\theta}})^{\intercal}(\boldsymbol{Y}-\boldsymbol{\Phi}_{G}\boldsymbol{W}\boldsymbol{\Xi}_{\boldsymbol{\theta}})) + \lambda_c\int_{\Omega}\boldsymbol{\xi}^{\intercal}_{\boldsymbol{\theta}}(\boldsymbol{v})\boldsymbol{W}^{\intercal}\boldsymbol{R}_{\gamma}\boldsymbol{W}\boldsymbol{\xi}_{\boldsymbol{\theta}}(\boldsymbol{v})d\boldsymbol{v}\right)
\end{equation}
where $\lambda_c>0$ is a penalty parameter determining the strength of prior regularization.
\par 
\newtext{The solution to \eqref{eqn:penalized_likelihood} is heavily dependent on $\lambda_c$, which balances data fit and distance to the prior. Its optimal value is difficult to set \emph{a-priori}, due to its sensitivity to specific data characteristics, e.g. $M$, $\sigma_e^2$, etc. Therefore, we propose an automated selection algorithm, using Bayesian optimization (BO) \citep{snoek2012}, to select $\lambda_c$ and (optionally) $\gamma$. See Section~\ref{apx:BO_hyperparam_tuning} of the supplemental material for more details.}
\par\bigskip 
\noindent{\textit{Estimating $\sigma_e^2$}}: In most diffusion experiments, we have access to $p>2$ non-diffusion weighted ($b=0$) images, which we denote as $\{\boldsymbol{y}_{i}^{0}\}_{i=1}^N$. Plugging $b=0$ into Equation~\ref{eqn:EAP} and assuming measurement error is independent of $b$-value, it follows that $\boldsymbol{y}_{i}^{0} \sim \mathcal{N}(\boldsymbol{1}_p,\sigma_e^2\boldsymbol{I}_p)$, and hence an estimate of $\sigma_e^2$ can be formed as
\begin{equation}\label{eqn:noise_estimator}
    \hat{\sigma}_e^2 = \frac{1}{N}\sum_{i=1}^N \widehat{\text{Var}}(\boldsymbol{y}_{i}^{0}),  
\end{equation}
where $\widehat{\text{Var}}$ denotes the empirical variance. 
\bigskip\par 
\noindent{\textit{Estimating $\sigma_{w}^2,\sigma_{\mu}^2$}}: We estimate the final model parameters $\sigma_{w}^2,\sigma_{\mu}^2$ by maximizing the posterior predictive distribution on a held-out calibration set. Specifically, for any $\boldsymbol{v}$, denote $\boldsymbol{g}_{\boldsymbol{v}} = (g_{\boldsymbol{v}}(p_1), ..., g_{\boldsymbol{v}}(p_M))$, then the conditional posterior predictive distribution for a new noisy data $\boldsymbol{y}=(y_1,...,y_{M})$ can formed as
\begin{equation}\label{eqn:local_calibration_likelihood}
\begin{aligned}
    &p(\boldsymbol{y}|\boldsymbol{v},\boldsymbol{Y},\boldsymbol{V},\text{rest}) = \int \underbrace{p(\boldsymbol{y}|\boldsymbol{v}, \boldsymbol{g}_{\boldsymbol{v}}, \text{rest})}_{\text{Likelihood}}\underbrace{p(\boldsymbol{g}_{\boldsymbol{v}}|\boldsymbol{Y},\boldsymbol{V},\text{rest})}_{\text{Equation~\ref{eqn:function_space_posterior}}}d\boldsymbol{g}_{\boldsymbol{v}} \\
    &= \mathcal{N}_{M}\big(\boldsymbol{\xi}_{\boldsymbol{\theta}}^{\intercal}(\boldsymbol{v}_{calib})\boldsymbol{\mu}\boldsymbol{1}_{\boldsymbol{M}} +
\frac{1}{\sigma_{e}^2}\boldsymbol{\Phi}_{G}[\boldsymbol{\xi}_{\boldsymbol{\theta}}^{\intercal}(\boldsymbol{v}_{calib})\otimes\boldsymbol{I}_{K}]\boldsymbol{\Lambda}_{\boldsymbol{\theta}}^{-1}[\boldsymbol{\Xi}_{\boldsymbol{\theta}}^{\intercal}\otimes\boldsymbol{\Phi}_{G}]^{\intercal}\text{vec}(\boldsymbol{Y}^{(c)}_{train}), \\ &\qquad \sigma_{\mu}^2\boldsymbol{1}_{\boldsymbol{M}}\boldsymbol{1}_{\boldsymbol{M}}^{\intercal} + \boldsymbol{\Phi}_{G}[\boldsymbol{\xi}_{\boldsymbol{\theta}}^{\intercal}(\boldsymbol{v}_{calib})\otimes\boldsymbol{I}_{K}]\boldsymbol{\Lambda}_{\boldsymbol{\theta}}^{-1}[\boldsymbol{\xi}_{\boldsymbol{\theta}}^{\intercal}(\boldsymbol{v}_{calib})\otimes\boldsymbol{I}_{K}]^{\intercal}\boldsymbol{\Phi}_{G}^{\intercal}  +
\sigma_e^2\boldsymbol{I}_{\boldsymbol{M}}\big)
\end{aligned}
\end{equation}
where \textit{rest} is shorthand for the remaining conditioning parameters $(\boldsymbol{\theta}, \boldsymbol{\mu}, \gamma, \sigma_w^2,  \sigma_e^2, \sigma_\mu^2)$. Denote $\boldsymbol{Y}_{calib}, \boldsymbol{V}_{calib}$ as held-out calibration data from some (relatively small) number of voxels $N_{calib}$. We estimate the unknown variance parameters by maximizing the likelihood of $\boldsymbol{Y}_{calib}$ under the predictive distribution \eqref{eqn:local_calibration_likelihood}:
\begin{equation}\label{eqn:calibration_likelihood_optimization}
\hat{\sigma}_{w}^2,\hat{\sigma}_{\mu}^2 = \underset{\sigma_\mu^{2},\sigma_{w}^2}{\text{argmax}}\prod_{i=1}^{N_{calib}}p(\boldsymbol{y}_{i} |\boldsymbol{v}_{i}, \boldsymbol{V}_{train}, \boldsymbol{Y}_{train}, \widehat{\boldsymbol{\theta}}, \widehat{\boldsymbol{\mu}}, \gamma, \sigma_w^2,\hat{\sigma}_e^2,\sigma_\mu^2),
\end{equation}
In both simulation and real data experiments, we find that maximizing Equation~\eqref{eqn:calibration_likelihood_optimization} over a relatively small 2-dimensional grid, the ranges of which are chosen via domain  knowledge, is sufficient for good performance, though alternative optimization or sampling approaches could also be applied.

\subsection{Uncertainty Quantification}\label{ssec:uncertianty_quantification}
We now discuss how to use the derived approximate predictive distribution for uncertainty quantification. Point-wise $100(1-\alpha)\%$ credible intervals for the field estimates at any $\boldsymbol{v}$ can be obtained in closed form using the quantiles of the normal distribution $Z_{1-\frac{\alpha}{2}}$ for any $\boldsymbol{p}\in\mathbb{S}^2$ via
\begin{equation}\label{eqn:pw_confidence_interval}
\widehat{\mathbb{E}}\left[g(\boldsymbol{v}, \boldsymbol{p})\right] \pm Z_{1-\frac{\alpha}{2}}\sqrt{\widehat{\text{Var}}\left[g(\boldsymbol{v},\boldsymbol{p})\right]},
\end{equation}
where $\widehat{\mathbb{E}}\left[g(\boldsymbol{v},\boldsymbol{p})\right], \widehat{\text{Var}}\left[g(\boldsymbol{v},\boldsymbol{p})\right]$ are shorthand for the mean and covariance function from  
\eqref{eqn:function_space_posterior}. 
\par 
For most applications, it is important to estimate some latent quantity of interest (QOI) from the ODF. For example, the ODFs can be used to form scalar statistics quantifying local structural features, which are then used in downstream tasks such as tissue segmentation, tractography and mass multivariate groupwise statistical analysis. Specifically, letting $T$ denote the function that maps the ODF $g_{\boldsymbol{v}}$ to some QOI $\tau_{\boldsymbol{v}}$, we can propagate the uncertainty through $T$ via sampling:
\begin{equation}\label{eqn:QOI_inference}
\begin{aligned}
    &g_{\boldsymbol{v}} \sim \text{ posterior in \eqref{eqn:function_space_posterior}} \\
&\tau_{\boldsymbol{v}} = T(g_{\boldsymbol{v}}),
\end{aligned}
\end{equation}
which can then be used for inference. In the special case that $T$ is a linear map, the distribution of $\tau_{\boldsymbol{v}}$ can be obtained in closed from. 

\section{Experiments}\label{sec:experiments}

\subsection{Datasets}

\subsubsection{Synthetic Phantoms}
A synthetic ODF field is generated using the multi-tensor model, with corresponding diffusion signal function defined as
\begin{equation}\label{eqn:mixture_tensor_model}
            S(\boldsymbol{v},\boldsymbol{p}|b) = \sum_{t=1}^T \frac{1}{T}\text{exp}(-b\boldsymbol{p}^\intercal \boldsymbol{D}_{t}(\boldsymbol{v})\boldsymbol{p})
\end{equation}
where $\boldsymbol{D}_{t}(\boldsymbol{v})$ is the $t$'th diffusion tensor at spatial location $\boldsymbol{v}$, $b$ is the b-value of the acquisition and $T$ is the number of fibers. In all of our experiments, we fix $b=3,000\text{s}/\text{mm}^2$ and assume a cylindrical model of diffusion with fixed eigenvalues  of $(15, 0.3, 0.3)\times 10^{-2}\text{mm}^2/\text{s}$. In order to avoid contaminating analysis with truncation bias, i.e. the irreducible bias resulting from representing the functions~\eqref{eqn:mixture_tensor_model} using a finite number of $K$ basis functions, we project the signal into $\text{span}(\boldsymbol{\phi})$ using linear regression over a dense sampling of function evaluations on $\mathbb{S}^2$. An electrostatic repulsion algorithm \citep{jones1999} is used to create uniformly sampled gradient directions on $\mathbb{S}^2$ for angular sample size $M$. The observed diffusion signal is simulated from measurement model \eqref{eqn:stat_model_observed_data} with measurement error variance $\sigma_e$, defining the theoretical SNR$:=1/\sigma_e$. 
\par 
We consider both a relatively simple 2D and more complex 3D geometry. For the former,  we simulate the diffusion field resulting from a crossing pattern of two fiber bundles perpendicular to one another in a 2D rectangular area. As depicted in the left panel of Figure~\ref{fig:2d_crossing_analysis_qualitative}, this geometry results in three distinct regions in the slice: a single fiber region with peak direction oriented along the x-axis, a single fiber region with peak direction oriented along the y-axis, and a two fiber region with a 90 degree crossing angle. For the latter, we construct a 3D phantom in the shape of a Caduceus. Specifically, the eigenvectors of the diffusion tensor field are created from the tangents of two intersecting 3D volumetric spirals, resulting in several single fiber regions along with a crossing fiber region with variable crossing angle, as displayed in Figure~\ref{fig:path_analysis}a. \newtext{All synthetic experiments are repeated for 50 Monte-Carlo replications.}

\subsubsection{\textit{In-vivo} Data}
A publicly available high-resolution (760 $\mu m^3$) dataset gathered on the MGH-USC 3 T Connectom scanner was used for evaluation. Full acquisition and processing can be found in \cite{wang2021}. Briefly, a single patient was scanned in 9 separate sessions and a total of 2,808 diffusion volumes were collected split between 420 $b=1,000s/mm^2$, 840 $b=2,500s/mm^2$ and 144 $b=0s/mm^2$, along with the corresponding reverse phase-encoding volumes. The total scanning time was approximately 18 hours. This dataset was chosen for analysis due to the low SNRs encountered, owing to the high spatial resolution of the acquisition.
\subsection{Implementation Details}\label{ssec:implementation_details}
The neural field used to parameterize $\boldsymbol{\xi}_{\boldsymbol{\theta}}$ has a SIREN architecture \citep{sitzmann2020}, consisting of an initial $\text{sine}$ encoding layer with random frequencies  followed by an $L$-layer multi-layer perceptron with $\text{sine}$ non-linearities. For maximizing~\eqref{eqn:penalized_likelihood}, we followed the initialization scheme discussed in \cite{sitzmann2020} and used the Adam optimizer with learning rate of $10^{-4}$. As $L$ and $r$ implicitly control the rank of the function representation space, we must adjust these parameters based on the complexity of the underlying field. For the 2D synthetic phantom, we set $L=3$ with $r=64$ and trained for 500 iterations. For the $3D$ phantom and 3D real-data ROI, we set $L=3$ with $r=128$ for 2,000 iterations. For the real-data full slice 2D ROI, we use  $L=3$ with $r=256$ and again train for 2,000 iterations. We initially selected $\lambda_c$ and $\gamma$ jointly in our BO-based hyperparameter optimization scheme, and found the performance to be relatively stable on a sub-manifold of the parameter space. As we have some domain knowledge on the properties of the ODFs for a given $b$-value, we ultimately found better performance by fixing $\gamma$ and selecting $\lambda_c$ using 20 iterations in Algorithm~\ref{alg:BO_hyperparameter_optimization} of the supplemental material. We set the length-scale parameter $\rho=0.5$ and took $\nu=1.0$ for $b=3,000$ and $\nu=2.0$ for $b=1,000$, as it is well known that the smoothness of the ODFs using approximation \eqref{eqn:FRT} is inversely related to the $b$-value. We used harmonics up to order 8, resulting in $K=45$ angular basis functions. 

\subsection{Competing Approaches}
We compare our methodology, from here on referred to as neural orientation distribution field (NODF), to popular alternative methods for ODF estimation and uncertainty quantification. For estimation, the roughness penalized spherical harmonic least squares estimator (SHLS) from \cite{descoteaux2007}, defined as 
\begin{equation}\label{eqn:ridge_regression}
        \hat{\boldsymbol{c}}_{\boldsymbol{v}_{i}} = \underset{\boldsymbol{c}}{\text{argmin}} \overset{M}{\underset{m=1}{\sum}} (y_{im} - \boldsymbol{c}^\intercal \boldsymbol{\phi}(\boldsymbol{p}_{m}))^2 + 
        \lambda \int_{\mathbb{S}^2} (\Delta_{\mathbb{S}^2}(\boldsymbol{c}^\intercal \boldsymbol{\phi}(\boldsymbol{p})))^2 d\boldsymbol{p},
\end{equation}
with penalty strength $\lambda > 0$, is used. As we are interested in investigating performance in low SNR regimes, we consider SHLS both applied to the raw data (SHLS-Raw) and the spatially smoothed data obtained from applying the local PCA method from \cite{veraart2016} (SHLS-MPPCA). The MPPCA-based approach was included in our comparisons as it is a common step for image analysis in noisy regimes.
\par 
As the overwhelming majority of existing methods to quantify the uncertainty in ODF estimation propose the use of some variant of the bootstrap, we adopt a residual bootstrap procedure for uncertainty quantification of the SHLS-Raw and SHLS-MPPCA estimates. \newtext{Please see Section~\ref{ssec:bootstrap_discussion} of the supplemental material for further details on this procedure.} We use the \textbf{dipy} implementations of the regression estimator~\eqref{eqn:ridge_regression} and MPPCA. For MPPCA, we use the suggested defaults and set the patch radius parameter to 2, resulting in a smoothing window of $5\times5\times 5$. The parameter $\lambda$ in Equation~\eqref{eqn:ridge_regression} was selected using generalized cross validation \citep{wahba1979}, which we implemented in python.
\subsection{Evaluation Metrics}

\subsubsection{ODF Evaluation}
To evaluate the estimation performance of the methods, we compute the (normalized) $L^2(\mathbb{S}^2)$ error:
$
\|\widehat{g}_{\boldsymbol{v}} - g_{\boldsymbol{v}}\|_{L^{2}}/\|g_{\boldsymbol{v}}\|_{L^{2}},
$
between the estimated ($\widehat{g}_{\boldsymbol{v}}$) and true ($g_{\boldsymbol{v}}$) ODF. 
To assess the uncertainty quantification, we compute the point-wise 95\% intervals for the underlying ODF at 200 (approximately) equispaced directions on $\mathbb{S}^2$. For NODF, these intervals can be formed using \eqref{eqn:pw_confidence_interval}. For the competitors, the residual bootstrap was applied to obtain $500$ bootstrapped estimates, which were then used to form bootstrapped point-wise 95\% intervals. The coverage was evaluated by computing the empirical coverage proportion (ECP), defined to be the proportion of the directions where the constructed interval contained the true function value. To assess the precision of the predictions, we also computed the average point-wise interval length (IL). For example, the solid colored surface in Figure~\ref{fig:ecp}a shows the true ODF from a randomly selected voxel in the crossing fiber region of the 2D phantom considered in Section~\ref{sssec:2d_phantom}, plotted in spherical coordinates. The translucent surfaces show the 95\% upper and lower confidence surfaces formed from Equation~\eqref{eqn:pw_confidence_interval}, for sparse ($M=10$) and dense ($M=60$) samplings. The ECP measures the proportion of the true ODF contained between the constructed surfaces, while the IL represents the average distance between them. \newtext{Hence, given two methods that both exhibit good coverage properties, e.g., ECP near $0.95$, we prefer a method that tends to result in smaller IL.}
\begin{figure}[!ht]
    \centering
    \includegraphics[scale=0.7]{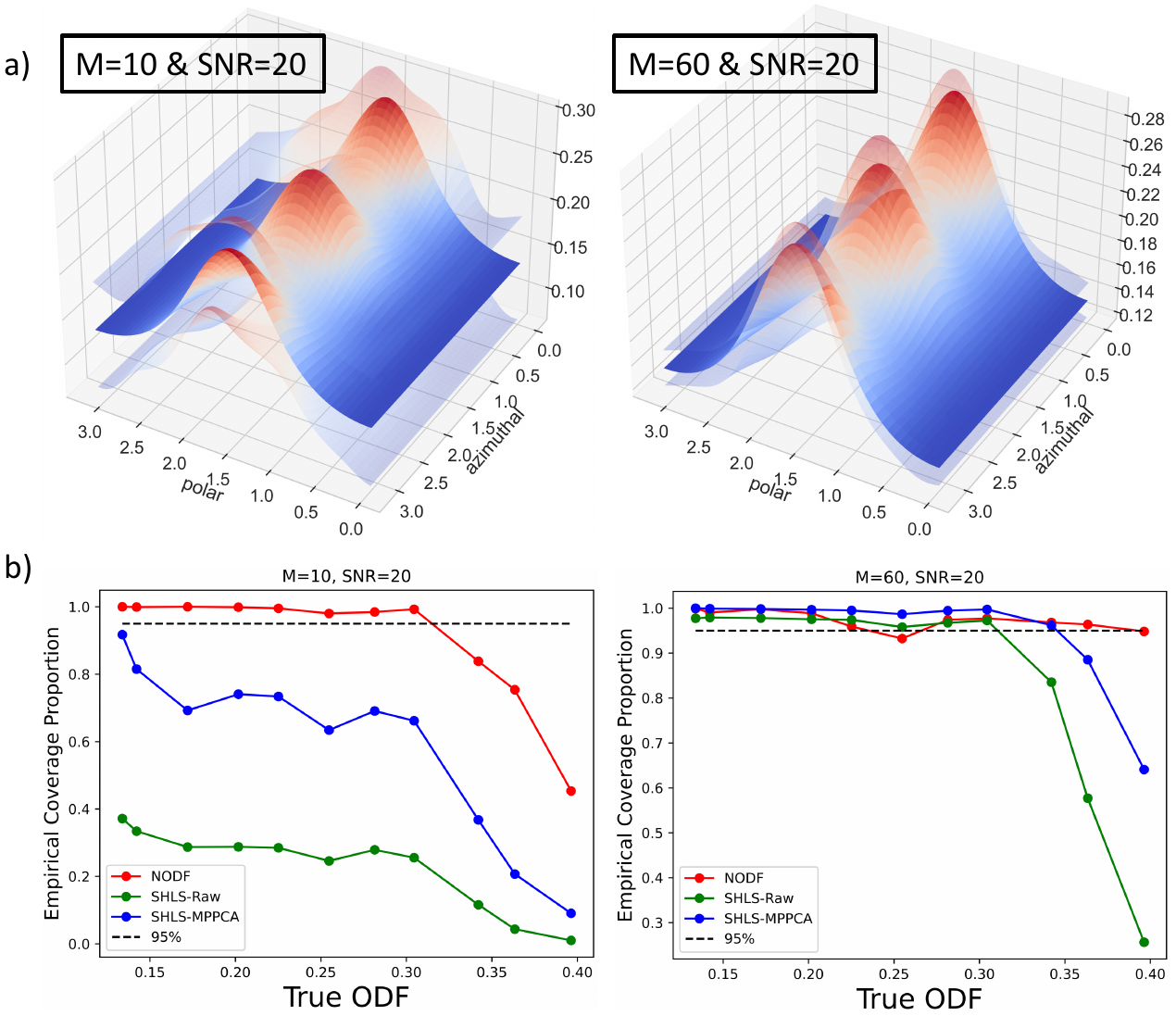}
    \caption{a) The solid color shows the true ODF in a randomly selected crossing voxel of the 2D phantom, plotted in spherical coordinates. The translucent surfaces show the upper and lower 95\% confidence surfaces formed via Equation~\eqref{eqn:pw_confidence_interval}. b) Monte Carlo average ECP plotted as a function of the true ODF value.}
    \label{fig:ecp}
\end{figure}
\subsubsection{Derived Quantities of Interest}\label{sssec:derived_QOI}
\newtext{We also assess the methods in terms of the downstream estimation and uncertainty calibration for both \textit{local} and \textit{global} QOI's computed from the ODFs. For a local (voxel specific) quantity, we consider the generalized fractional anisotropy (GFA), defined as 
$$
\text{GFA}(h) = \sqrt{\frac{n \sum_{j=1}^n(h(\boldsymbol{p}_j) - n^{-1}\sum_{l=1}^{n}h(\boldsymbol{p}_l))^2}{(n-1)\sum_{j=1}^nh(\boldsymbol{p}_j)^2}} \quad h\in\mathcal{H},
$$
using some dense discretization $\{\boldsymbol{p}_1,...,\boldsymbol{p}_n\}\subset\mathbb{S}^2$. Since the GFA is a non-linear function of the ODF, we must perform inference via sampling \eqref{eqn:QOI_inference}. The performance is evaluated by computing the ECP and IL, along with the bias and absolute error, of the GFA.}
\par 
\newtext{For a global quantity, we consider the problem of inferring large-scale white matter fiber tracts from the diffusion data, i.e., \textit{tractography}. A simple approach to tractography is to model a tract as a 3D curve $x\in \mathcal{C}^{1}([0,1]\mapsto\Omega\subset\mathbb{R}^3)$ which is defined to be the solution to the initial value problem:
\begin{equation}\label{eqn:IVP}
    \frac{\partial x(t)}{\partial t} = u(x(t))\qquad x_0 = x(0) \in \Omega,
\end{equation}
where $u$ is a continuous vector field on $\Omega$ which models the principal diffusion directions. A simple peak detection algorithm is used to identify the principal diffusion directions as local ODF maxima ($\ge$ half the global maximum) over 2,562 points on an icosphere. Streamlines are then generated via Euler's method (step size 0.05) and terminating at points with fractional anisotropy $<0.25$. A ground truth streamline $x^{(gt)}$ for a seed $x_{0}$ is defined by the solution to \eqref{eqn:IVP} using the ground truth ODF field.}
\par 
\newtext{In practice, the observed data is used to form an estimate of the field $\widehat{u}$ by obtaining the principal diffusion directions from the estimated ODF, which is then plugged into Equation~\eqref{eqn:IVP} to estimate the streamlines. Propagating the uncertainty in the estimated ODFs to tract uncertainty can be accomplished by: i) (re)sampling the ODFs (using posterior~\eqref{eqn:function_space_posterior} for NODF and bootstrap for the competitors), ii) calculating the principal diffusion directions, iii) generating curves according Equation~\eqref{eqn:IVP} for each sample.}
\par 
\newtext{From the resulting sample of curves, we would like to evaluate the performance of the probabilistic tractography as well as obtain some quantification of the uncertainty.
The performance was quantified by computing the minimum $L^2([0,1]\mapsto \Omega)$ distance between the ground truth $x^{(gt)}$ and 10 \textit{deepest} curves in the sample, as quantified by the curve depth measure from \cite{micheaux2021}. The depth measure can be considered as an extension of the univariate order statistics, with the deepest curve in the sample an analog of the sample median, and multiple deep curves are used in evaluation due to the often encountered multi-modality of the sampled paths (see Figure~\ref{fig:path_analysis}b for an example). Uncertainty was quantified by computing the cross sectional angular dispersion (AD) measure of the sample. Specifically, for any $t$, we can compute the point-wise empirical covariance matrix of the directions 
$$
\boldsymbol{A}(t) = \frac{1}{n}\sum_{i=1}^n\frac{\partial_t x_i(t)}{\|\partial_t x_i(t)\|_2}\left[\frac{\partial_t x_i(t)}{\|\partial_t x_i(t)\|_2}\right]^{\intercal},
$$
where $x_i$ is a sampled path and $n$ is the number of samples. The eigenvalues of $\boldsymbol{A}(t)$ can be used to form the angular dispersion (AD) measure, defined as 
$$
\text{AD}(t):=\text{sin}^{-1}(\sqrt{ 1 - \text{EigMax}(\boldsymbol{A}(t))})$$ 
\citep{schwartzman2008}. Notice that $\text{AD}(t)=0$ if the derivative of all curves in the sample point in same direction, and approaches the maximum value $\approx 0.94$ radians when the directions are uniformly distributed on the $\mathbb{S}^2$. To handle the varying length, the curves were re-sampled using cubic b-splines over a uniform grid of 100 points on $[0,1]$.}

\subsection{Results}

\subsubsection{2D Phantom}\label{sssec:2d_phantom}

\begin{figure}[!ht]
    \centering
    \includegraphics[width=\textwidth]{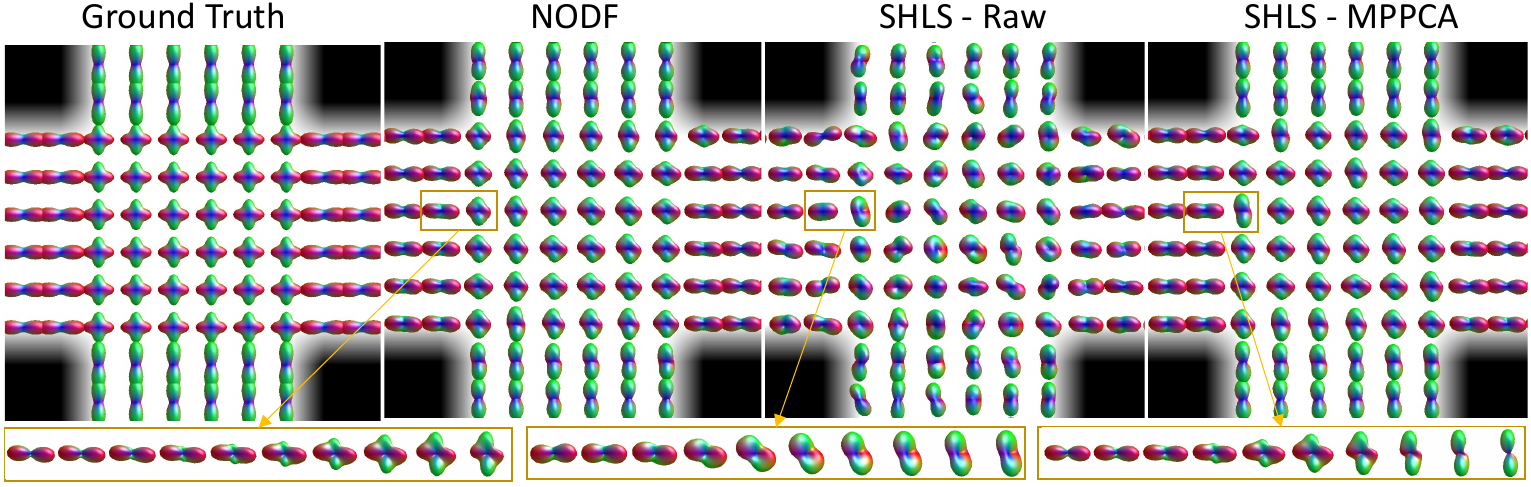}
    \caption{True ODF field and estimates for $M=10$. \newtext{The yellow boxes show off-grid predictions formed on an equispaced grid between adjacent voxel coordinates using the predictive mean in Equation~\eqref{eqn:function_space_posterior} for NODF, and bilinear interpolation of the coefficients for SHLS-Raw and SHLS-MPPCA.} \newtext{Min-max normalization was applied to the interpolants to highlight orientation details.}}
\label{fig:2d_crossing_analysis_qualitative}
\end{figure}
\begin{table}[!ht]
    \centering
    \small
\begin{tabular}{llrrrrrr}
\toprule
    &    & \multicolumn{3}{l}{M=10} & \multicolumn{3}{l}{M=60} \\
    &    &    NODF & SHLS-Raw & SHLS-MPPCA &    NODF & SHLS-Raw & SHLS-MPPCA \\
 &  &         &          &            &         &          &            \\
\midrule
ODF & ECP &  0.9819 &   0.2925 &     0.7280 &  0.9832 &   0.9556 &     0.9898 \\
    & IL &  0.1071 &   0.0206 &     0.0467 &  0.0291 &   0.0816 &     0.0343 \\
    & $L^2$-Error &  0.0956 &   0.1429 &     0.0994 &  0.0243 &   0.0690 &     0.0263 \\
GFA & ECP &  0.9813 &   0.2540 &     0.3521 &  0.9510 &   0.9663 &     0.9658 \\
    & IL &  0.1284 &   0.0354 &     0.0697 &  0.0422 &   0.1236 &     0.0485 \\
    & Abs. Error &  0.0261 &   0.0512 &     0.0447 &  0.0121 &   0.0282 &     0.0127 \\
    & Bias &  0.0047 &  -0.0330 &    -0.0401 &  0.0010 &  -0.0182 &    -0.0100 \\
\bottomrule
\end{tabular}
    \caption{Simulation results for 2-D crossing phantom. \newtext{Associated standard errors can be found in Table~\ref{tab:sim_2d_results_standard_errors} of the supplement.}}
\label{tab:2d_crossing_analysis_qualitative}
\end{table}
\newtext{The Monte-Carlo averaged results in Table~\ref{tab:2d_crossing_analysis_qualitative} (standard errors in Table~\ref{tab:sim_2d_results_standard_errors} of the supplemental materials) compare the methods' performance in estimation and uncertainty quantification of both the latent ODF field and GFA for the 2D synthetic data under extremely sparse $M=10$ and relatively dense $M=60$ regimes, with $\text{SNR}=20$.} The $L^2$ errors indicate that SHLS-Raw does not produce reasonable estimates in the super sparse ($M$=10) case. 
This is unsurprising, as SHLS-Raw does not integrate spatial information to improve estimation efficiency, relying only on an angular smoothness prior (i.e. the Laplacian-based roughness penalty), which is  insufficient in this regime. Table~\ref{tab:2d_crossing_analysis_qualitative} also shows severe under-coverage for the point-wise intervals formed from the residual bootstrap for $M=10$. Coupled with the small IL, we see that the uncertainty is being dramatically underestimated in this sparse regime.
Leveraging the local neighborhood information by first denoising the data with MPPCA and then applying the SHLS significantly improves the point estimates for both sparse and dense samplings, as is reflected in the $L^2$ errors in Table~\ref{tab:2d_crossing_analysis_qualitative}. However, the bootstrapped point-wise intervals from SHLS-MPPCA also exhibit significant under-coverage for the $M=10$ case. In contrast, we see that the NODF estimates have the lowest average $L^2$ error with the ground truth fields for both $M=10$ and $M=60$, while also maintaining properly calibrated point-wise intervals (on average) in both regimes. 
\newtext{This is particularly notable in the $M=10$ regime, where the large IL for NODF accurately reflects the uncertainty inherent in this sparse sample and noisy set-up.} Furthermore, in the $M=60$ regime where all methods produce point-wise intervals that are calibrated on average, NODF displays the smallest IL, indicating the tightest intervals between all the methods. 
\par 
Figure~\ref{fig:ecp}b plots the ECP as a function of the true underlying ODF value, averaged over all voxels. For the sparse $M=10$ case, our method displays dramatically better coverage than the alternative approaches for all ODF values, achieving the desired $\ge0.95 $ for all but the peak ODF values, where all approaches exhibit deteriorated performance, likely due to the bias induced by the need for strong prior regularization in this regime. For the $M=60$ regime, our method is well calibrated for all the ODF values, including the peak, in which both SHLS-Raw and SHLS-MPPCA bootstrap continue to exhibit under-coverage. For the SHLS-Raw bootstrap method, the observed under-coverage at the peak could potentially be mitigated by undersmoothing the data, i.e. manually setting 
$\lambda$ in Equation~\eqref{eqn:ridge_regression} to be very small or 0. That said, due to the bias-variance trade-off, this would significantly increase the variability of the estimates, further increasing the IL which, as we see in Table~\ref{tab:2d_crossing_analysis_qualitative}, is already much larger than NODF, and likely leading to overall deterioration in the $L^2$-error. How to avoid this effect in the SHLS-MPPCA bootstrap is less clear, as taking the equivalent bias-reduction approach for the local PCA would simply be to chose a large rank, which would effectively just return the raw data and reduce the procedure to SHLS-Raw.
\par 
For estimation and uncertainty quantification of the GFA, Table~\ref{tab:2d_crossing_analysis_qualitative} shows that sampling from the NODF posterior yields the lowest average bias and demonstrates proper calibration in terms of the ECP achieving the nominal level, for both sparse and dense set-ups. \newtext{Notably, in the M=10 regime, though the average ODF $L^2$ errors from SHLS-MPPCA are less than 5\% larger than those of NODF, the gap in GFA inference is much more pronounced. Specifically, the GFA error from SHLS-MPPCA is nearly double that of NODF. SHLS-MPPCA also shows bias that is an order of magnitude larger than NODF, and we observe further deterioration in the ECP.}
The point estimates dramatically improve for SHLS-MPPCA in the $M=60$ case, but the NODF estimates still lead to less bias along with tighter intervals and lower error. 
\par 
A thus far understated advantage of the proposed methodology is the native ability to evaluate the posterior predictive distribution of the field at any spatial location $\boldsymbol{v}\in\Omega$. That is, once trained, predictive inference can be formed continuously for the entire domain.
\newtext{The yellow boxes in Figure~\ref{fig:2d_crossing_analysis_qualitative} show the estimates over an equispaced grid between two neighboring voxel coordinates, formed using the posterior mean for NODF and bilinear interpolation on the spherical harmonic coefficients for both SHLS-Raw and SHLS-MPPCA. For NODF,} we observe that at all points along the path, the ``off the grid'' field estimates are a reasonable mixture between the single and crossing fiber regions. 

\subsubsection{3D Phantom}\label{sssec:phatom3D} 
\begin{table}[]
    \centering
    \addtolength{\leftskip} {-2cm}
    \addtolength{\rightskip}{-2cm}
    \small
\begin{tabular}{ll|rrr|rrr|rrr}
\toprule
   &    & \multicolumn{3}{c}{NODF} & \multicolumn{3}{c}{SHLS-Raw} & \multicolumn{3}{c}{SHLS-MPPCA} \\
M   &  SNR  & $L^2$ & ECP & IL & $L^2$ & ECP & IL  & $L^2$ & ECP & IL \\
\midrule
10 & 20 &        0.108 &          0.998 &        0.159 &        0.214 &          0.233 &        0.014 &        0.119 &          0.827 &        0.049 \\
   & 10 &        0.145 &          0.997 &        0.166 &        0.354 &          0.851 &        0.129 &        0.172 &            0.864 &        0.095 \\
20 & 20 &        0.084 &          0.993 &        0.106 &        0.163 &          0.952 &        0.091 &        0.089 &         0.866 &        0.038 \\
   & 10 &        0.116 &          0.987 &        0.111 &        0.280 &          0.973 &        0.159 &        0.129 &           0.903 &        0.073 \\
30 & 20 &        0.052 &          0.980 &        0.056 &        0.134 &          0.974 &        0.088 &        0.060 &          0.925 &        0.033 \\
   & 10 &        0.086 &          0.960 &        0.057 &        0.238 &          0.991 &        0.159 &        0.099 &           0.926 &        0.062 \\
40 & 20 &        0.048 &          0.973 &        0.040 &        0.119 &          0.985 &        0.088 &        0.054 &           0.937 &        0.030 \\
   & 10 &        0.076 &          0.942 &        0.041 &        0.214 &          0.997 &        0.165 &        0.088 &         0.935 &        0.055 \\
50 & 20 &        0.045 &          0.973 &        0.037 &        0.108 &          0.991 &        0.088 &        0.049 &         0.958 &        0.030 \\
   & 10 &        0.070 &          0.942 &        0.037 &        0.196 &          0.999 &        0.167 &        0.080 &         0.969 &        0.062 \\
60 & 20 &        0.041 &          0.973 &        0.035 &        0.101 &          0.994 &        0.088 &        0.046 &        0.967 &        0.030 \\
   & 10 &        0.066 &          0.948 &        0.037 &        0.184 &          1.000 &        0.169 &        0.075 &       0.973 &        0.058 \\
\bottomrule
\end{tabular}
\caption{Monte-Carlo averaged results for the $L^2(\mathbb{S}^2)$ error ($L^2$), point-wise interval length (IL) and empirical coverage proportion (ECP) over the 3D synthetic phantom. \newtext{Associated standard errors can be found in Table~\ref{tab:sim_3d_results_standard_errors} of the supplement.}}
    \label{tab:sim_3d_pointwise_results}
\end{table}
\par\bigskip
\noindent{\textit{Local Error Analysis and Uncertainty Calibration}: The Monte-Carlo simulation averaged results over all voxels in the phantom are displayed in Table~\ref{tab:sim_3d_pointwise_results} (standard errors in Table~\ref{tab:sim_3d_results_standard_errors} of the supplemental materials) \newtext{for various angular sample sizes $M$, for both high and low SNR (20 and 10, respectively)}. NODF uniformly outperforms the competitors in terms of average $L^2(\mathbb{S}^2)$ error for all settings. Additionally, we see that the ECP of our method is generally in line with the $0.05$ level for all $M$ and SNR considered. As observed in the 2D phantom results, we notice poor empirical coverage of the bootstrap methods for the super sparse $M=10$ case. The empirical coverage of the SHLS-Raw bootstrap is well calibrated as $M$ increases, but at the cost of point-wise intervals which are very large, indicating substantial uncertainty in the predictions. Echoing the results from Section~\ref{sssec:2d_phantom}, the SHLS-MPPCA bootstrap results in reasonably good point estimates, but we notice the under-coverage persist for larger $M$ as compared to SHLS-Raw. 
Notice that for all cases, NODF follows the expected trend of higher uncertainty (IL) for fewer sample directions or low SNR. Additionally, for a fixed $M$, the IL of our method is more robust to the low SNR case than the competitors, i.e. indicating more precision in the ODF estimates for noisy regimes.
\par\bigskip
\noindent{\textit{Path Analysis and Global Uncertainty Propagation}}: \newtext{To evaluate the methods' performance in tractography, we obtain a set of ground truth streamlines $x^{(gt)}$ using the ground truth field for 100 randomly selected seeds $x_{0}$ in the $z=0$ plane. 100 posterior samples for our method and 100 bootstrapped samples for the competitors were formed using the approach outlined in Section~\ref{sssec:derived_QOI} and used for subsequent analysis.}
\par 
The top row of Figure~\ref{fig:path_analysis}b shows an example ground truth streamline (black dashed curve) from a randomly selected seed, along with the posterior streamlines for the methods in both sparse $M=20$ and dense $M=60$ cases. The black dotted curve shows the center-line of the crossing tube. The large volume of the crossing region, existence of small crossing angles and low SNR used for the simulation makes this a difficult tracing problem. For the sparse $M=20$ case, we see that only our method is able to produce posterior streamlines that follow the true streamline, though the spatial spread of the curves is large and we see some samples jumping to the crossing streamline. SHLS-Raw also has a high degree of uncertainty, but sampled streamlines fail to cover the true path. Echoing results from the local analysis, bootstrapping the SHLS-MPPCA for sparse samples underestimates the variability in possible paths, clustering tightly around the wrong turn path. As $M$ increases, the spread of the NODF paths tightens around the true streamline, converging faster than either of the two competitor methods. The bottom two rows of Figure~\ref{fig:path_analysis}b show the distribution of errors and angular dispersion of the posterior samples over all 100 seeds. Echoing the observations from the single randomly selected seed, we note lower errors for all $M$ as well as monotonically reducing uncertainties for our method. 
\begin{figure}[t]
    \centering
    \includegraphics[scale=0.7]{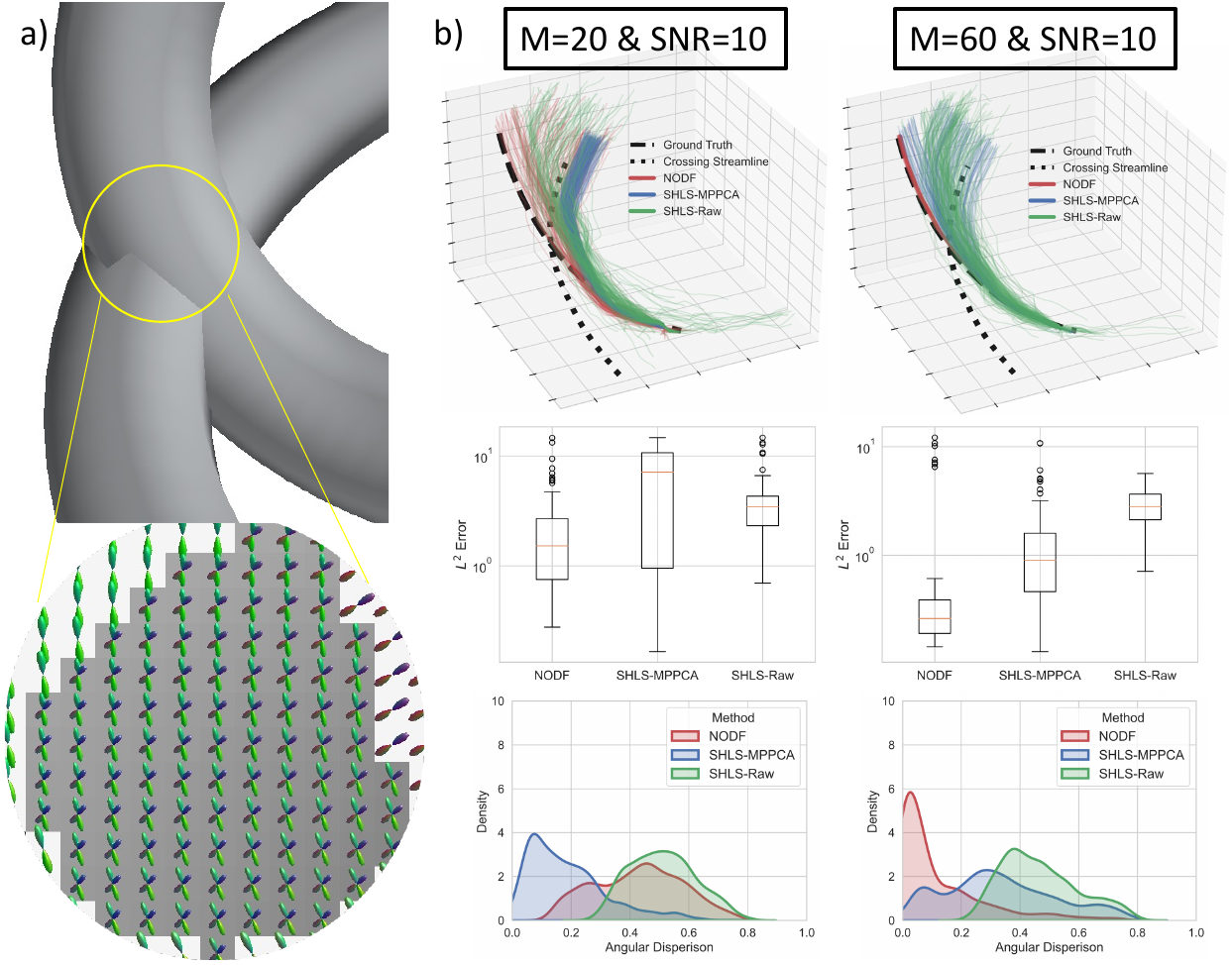}
    \caption{a) \newtext{(Top): 3D synthetic phantom. (Bottom) Slice of the 3D crossing region with background colored by FA. Pictured ODFs were sharped to enhance crossing information.} b) (Top row) Probabilistic tracing results from the 3D phantom for each method for sparse and dense sampling budgets $M$ for a fixed initial starting point. (middle row) $L^2$-curve errors (bottom row) and angular dispersion averaged over all starting points.}
    \label{fig:path_analysis}
\end{figure}

\subsubsection{Real Data Experiments}\label{ssec:real_data}
\begin{figure}[!ht]
    \centering
    \includegraphics[scale=0.6]{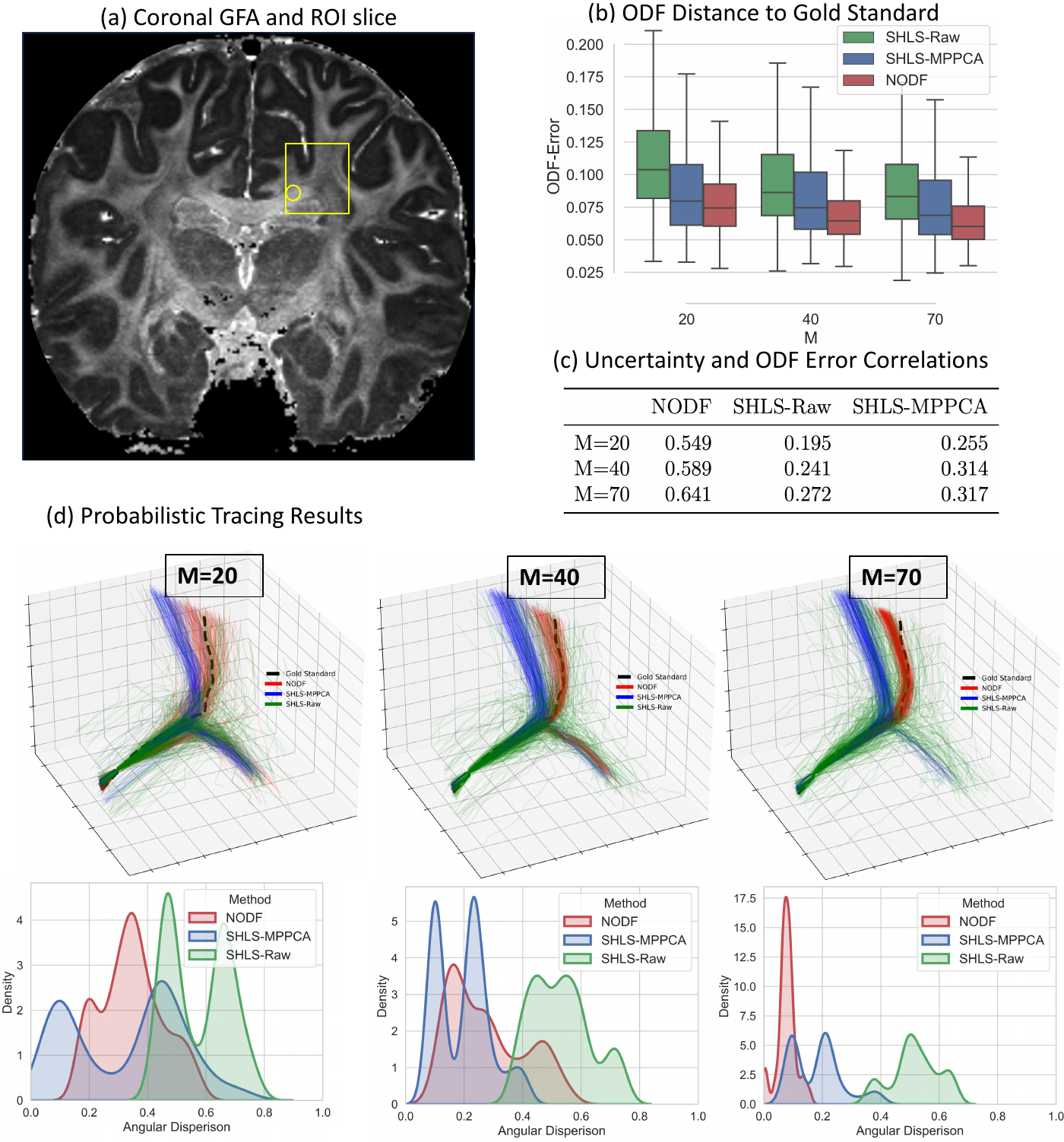}
    \caption{Real data application: (a) Coronal view of GFA map with 3D ROI cross section outlined in yellow. (b) Distribution of $L^2$-distances between estimated and gold standard ODFs across ROI. (c) Correlation between $\text{CV}_{\text{GFA},\boldsymbol{v}}$ and $L^2$-distances to gold standard. (d) Probabilistic tracing results for initial seed in corpus callosum, identified by yellow circle in panel (a).}
    \label{fig:ROI_3D_results}
\end{figure}
\newtext{We now evaluate the methods' performance on the high-resolution in-vivo dataset.} Since we do not have a ground truth for in-vivo data, we apply SHLS to all 9 session's data at $b=1,000s/mm^2$. The relatively large number of gradient directions per scan coupled with averaging over all 9 scans allows reasonable performance for this estimator. The resulting ODFs are used as a ``gold standard'' for comparison, though there remains non-trivial uncertainty in the estimates. We then apply each method to the session 1 data for $M=20, 40$ and $70$ unique directions, and compare the recovery to the gold standard.
\par 
We begin our analysis on a 3D ROI of dimension 
$21\times 10 \times 15$, which covers parts of the corpus collosum, centrum semiovale and corticospinal tract. A 2D coronal cross section of this region is outlined in the yellow box in Figure~\ref{fig:ROI_3D_results}a. The box-plots in Figure~\ref{fig:ROI_3D_results}b show the distribution of normalized $L^2(\mathbb{S}^2)$ distances between the gold standard ODFs and the ODFs estimated from each method as a function of $M$ over the ROI. We see that for all budgets considered, our method predicts ODFs with the lowest distance relative to the gold standard data. 
\par 
Since we lack access to the ground truth parameter, assessing the uncertainty performance via the empirical coverage measures, \newtext{as was done in the synthetic data experiments}, is not feasible. Instead, we form a scalar summary of the uncertainty in the predicted ODF using the coefficient of variation of the generalized fractional anisotropy, defined as:
$\text{CV}_{\text{GFA},\boldsymbol{v}} = \text{Var}\left[\text{GFA}(g_{\boldsymbol{v}})\right]^{1/2}\big/\mathbb{E}\left[\text{GFA}(g_{\boldsymbol{v}})\right]$, approximated using posterior samples from \eqref{eqn:function_space_posterior} for our method and the residual bootstrap for the alternatives. As it is common practice in medical imaging to utilize scalarized uncertainty metrics as features in downstream prediction tasks \citep{mehta2022}, inverse weightings for groupwise inference procedures \citep{jens2018}, or as a surrogate for the unknown prediction accuracy \citep{tanno2021}, a  scalar uncertainty metric that is more highly correlated with the true normalized $L^2(\mathbb{S}^2)$ prediction errors
implies better uncertainty modeling. Figure~\ref{fig:ROI_3D_results}c displays the correlation between $\text{CV}_{\text{GFA},\boldsymbol{v}}$ and the ODF distance to gold standard. We see that for all $M$, our method produces uncertainty estimates with higher correlations.
\begin{figure}[!ht]
    \centering
    \includegraphics[scale=0.5]{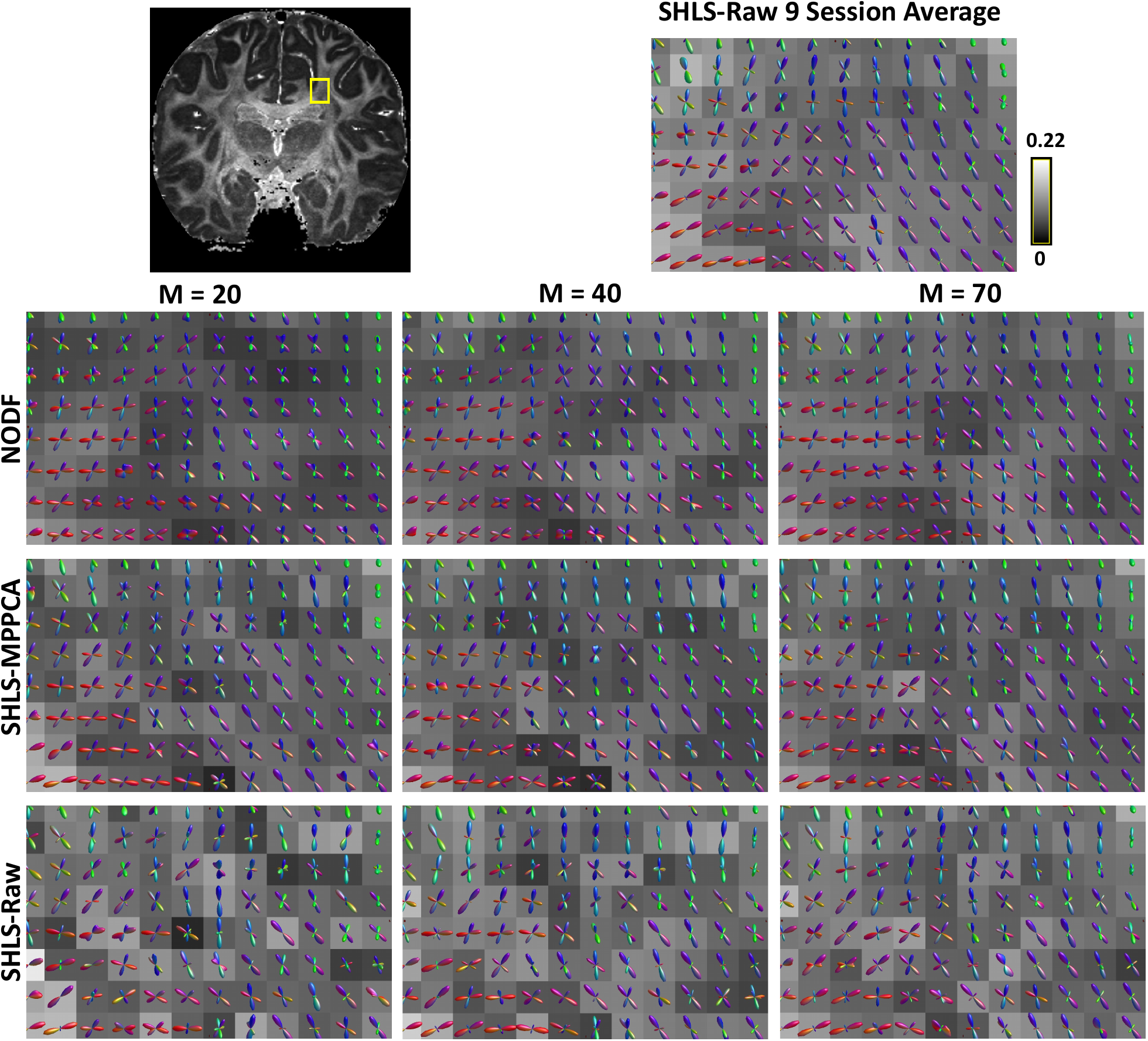}
    \caption{Field estimates for a \newtext{small} coronal slice of the 3D ROI \newtext{(yellow rectangle)} for each method and number of angular samples considered. Spherical deconvolution was applied to sharpen the directional information in the ODF estimates. \newtext{The background is colored using the mean GFA estimates for each method.}}
    \label{fig:roi3d_odf_estimates}
\end{figure}
\par 
Figure~\ref{fig:ROI_3D_results}d shows the probabilistic tractography paths from sampling \eqref{eqn:IVP} with initial seed in the corpus collosum, indicated by the yellow circle in Figure~\ref{fig:ROI_3D_results}a, using the same tracing algorithm and parameters outlined in \newtext{Section~\ref{sssec:derived_QOI}}. The gold standard curve is obtained by tracing the path from the gold standard ODFs. Results echo what was observed in the simulations. Specifically, NODF is the only method which produces streamlines which cover the gold standard curve for all $M$. The tract uncertainty measure (quantified using angular dispersion) is decreasing for NODF as a function of $M$, as the streamlines concentrate spatially near to the gold standard curve. In contrast, for $M=20$, SHLS-MPPCA under-estimates the uncertainty in the possible paths, producing two tight clusters, neither of which cover the true path, with many traces incorrectly taking a ``wrong turn'' down the corticospinal tract. For $M=70$, the SHLS-MPPCA tracings are somewhat reasonable, but are tightly concentrated and spatially biased with respect to the gold standard tracing. SHLS-Raw tracings are highly variable for all $M$, with many resulting in very improbable paths that circle around back down the corpus collosum. \newtext{Finally, it is important to acknowledge that uncertainties in ground truth ODF estimation, imprecise peak detection, and limited image resolution make it difficult to accurately determine if the fiber truly extends upwards or continues on laterally.}   
\par
The top \newtext{right} plot in Figure~\ref{fig:roi3d_odf_estimates} shows the ODF estimates from SHLS-Raw on the full 9 session data within the 2D cross section outlined in yellow in \newtext{the top left image}. The remaining panels show the mean field estimates from each method for each angular sample size on the session 1 data. The background of all plots is colored using the estimated \newtext{GFA}. For visual clarity, we sharpen the directional information by applying spherical deconvolution \citep{tournier2008resolving} to the estimated ODFs from each method. In comparison to the competitors, our method is able to infer ODF fields with far more spatial coherence in the directional information. 
\begin{figure}[!ht]
    \centering
    \includegraphics[width=\textwidth]{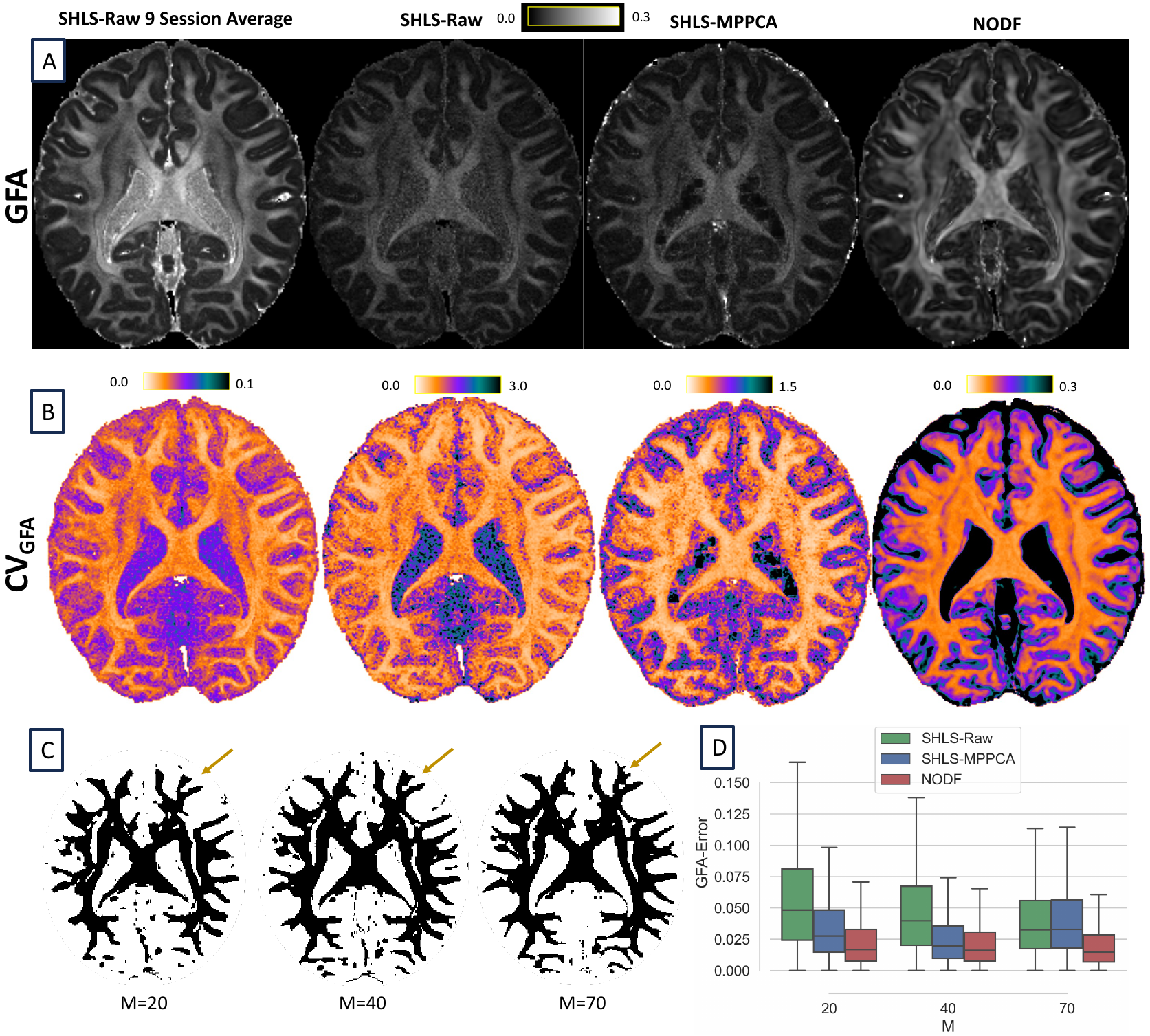}
    \caption{A) GFA estimates from an axial slice with $M=70$ directions. The corresponding GFA coefficient of variation is shown in B). C) The sets $\{\boldsymbol{v}\in\boldsymbol{V}:\text{GFA}\left(g_{\boldsymbol{v}}\right) > 0.05\}$ \newtext{(black)} estimated from the NODF posterior, with positive detections controlled voxel-wise at the 0.01 level. The boxplots in D) show the distribution of absolute distances between gold standard GFA and GFA estimates from session 1 data as a function of $M$.}
    \label{fig:gfa_uq}
\end{figure}
\par 
We also apply each method to a full 2D axial slice and analyze the performance on downstream estimation and uncertainty quantification of the GFA. Figure~\ref{fig:gfa_uq}a compares the GFA obtained from SHLS-Raw on the 9 session averaged data with the estimates from each method on the session 1 data for $M=70$ directions. The corresponding scalar uncertainty maps are provided below each image. We see that the uncertainty for NODF is generally lower along identifiable white matter structures and outlines sharp boundaries between white and gray matter. The boxplots in Figure~\ref{fig:gfa_uq}d display the distribution of absolute GFA errors (vs. session 1) for each method for all $M$ across the whole axial slice. We see that the GFA estimates from our method are generally the closest to the ``gold standard'' estimates. 
\par 
We note the presence of non-trivial uncertainties in the ``gold standard'' GFA (left plot of Figure~\ref{fig:gfa_uq}b), with notable overestimation occurring in specific areas of the ventricle regions, as depicted in the left plot Figure~\ref{fig:gfa_uq}a. This emphasizes that even when estimating relatively straightforward microstructural summaries like GFA in dense sampling regimes with standard estimation procedures, significant uncertainties can arise.
\par 
It may also be of interest to perform explicit statistical tests for various diffusion parameters. As a simple example, we may want to detect the set of voxels with: $\{\boldsymbol{v}\in\boldsymbol{V}:\text{GFA}\left(g_{\boldsymbol{v}}\right) > \mu_{\text{GFA}}\}$, for some threshold $\mu_{\text{GFA}} > 0$, e.g. for tissue segmentation. To do so, we form a voxel-wise $t$-statistic using samples from the NODF posterior \eqref{eqn:function_space_posterior} and identify significant voxels controlling (locally) for type-1 error at the 0.01 level. Figure~\ref{fig:gfa_uq}c shows the resulting significant region for $\mu_{\text{GFA}}=0.05$, with selected voxels \newtext{shown in black}. We note a clear recovery of white matter regions for all $M$ considered. Additionally, as $M$ increases, the increasing power allows for the detection of finer detailed structures, e.g. \newtext{the indicated area} near the white-grey interface. 

\section{Discussion and Future Work}\label{sec:discussion}
Measurement errors and the often ill-posed nature of the ODF inverse problem require some form of regularization/incorporating priors to allow for stable estimation. The key to our method is in the ability to learn a data-driven basis, parameterized by a flexible deep neural field, allowing implicit modeling of the spatial correlations in the data, which can then be leveraged to facilitate more powerful inference. In comparison to popular alternatives, the analysis presented suggest that NODF offers both better recovery of the ODFs, along with less bias and more reliable propagation of uncertainties to some downstream quantities of interest, under a variety of sampling regimes and noise levels.
\par 
\newtext{This methodology can also offer improvements in computational efficiency. Once the parameters of the distribution \eqref{eqn:function_space_posterior} are estimated, estimation and uncertainty quantification can be formed continuously at any spatial location in the image, requiring only a forward pass through the NF, along with some (relatively small) matrix computations. This is in contrast to resampling-based approaches, which can be computationally costly for all but the fastest estimators, owing to the large number of locations (voxels) we wish to make inference at in a typical imaging application. Please see Section~\ref{ssec:computation_bootstrap} of the supplemental material for further discussion on this topic.}
\par 
This method can be extended in several interesting directions. Most immediately, extensions to multiple $b$-value acquisitions is of interest. A simple way to apply the current method ``out of the box'' in such situations is fitting a separate field model to each shell. Due to our method's strong performance in the sparse angular sampling case, such an approach would be applicable even in the case of relatively few measurements at each $b$-value. Still, this ignores valuable correlations between the signal on different shells. Addressing the multi-shell situation in a more principled manner could involve replacing the Funk Radon transform with an alternative forward operator that is well defined for multi-shell data. However, further investigation is required to assess the feasibility and effectiveness of adopting such a strategy.
\par 
Improvements to the model and learning procedure for the spatial basis $\boldsymbol{\xi}_{\boldsymbol{\theta}}$ will be explored in the future. First, as our method is agnostic to the specific NF parameterization of $\boldsymbol{\xi}_{\boldsymbol{\theta}}$, alternative models can be seamlessly integrated. 
Due to the global nature of the sine activations, the SIREN parameterization used here may not be the best representation for the spatial basis system. Neural field architectures with spatially compact activation functions such as wavelets have shown promising performance in medical imaging inverse problems \citep{saragadam2023} and so may provide a better representation here. Second, in addition to the angular prior regularization, it may be beneficial to incorporate a spatial prior regularization term to the optimization problem \eqref{eqn:penalized_likelihood}, e.g. total variation. If purely orientational inference is desired, it may also be helpful to incorporate a sparsity inducing angular prior on the coefficients, e.g. Laplace ($l1$). However, doing so directly would destroy the Gaussian-Gaussian conjugacy that facilitates the fast uncertainty quantification, and hence we would suggest proceeding in a two stage manner, e.g. as was done by applying the second stage deconvolution to obtain the fODFs in Section~\ref{ssec:real_data}. \newtext{Finally, we find that to achieve good performance when fitting large-scale images, such as whole brain volumes, 
a large $L$ and $r$ are necessary (see experimental results provided in Section~\ref{ssec:abcd_subject} of the supplemental materials).  However, these experiments evaluate fits using data at moderate resolutions and signal-to-noise ratios more commonly encountered in diffusion MRI. Identifying the optimal approach to scale the method to full-brain fits from modern enormous ultra-high resolution research scans is an important direction for future research. An immediate straightforward approach is to simply fit multiple NODFs over some overlapping partition of the domain and then combine predictions \citep{tancik2022blocknerf}, though more sophisticated methods, e.g., utilizing hash grid encoding \citep{mueller2022}, can be explored.}
\par
Another important avenue for future research involves investigating the effect of the distribution of the measurement errors. While the additive normal model \eqref{eqn:stat_model_observed_data} for the noise is widely used in the statistical analysis of diffusion MRI and has been found to be quite reliable when the SNR $>3$ \citep{hakon1995}, it is often only an approximation to the true Rician/non-central $\chi^2$ noise distribution. It has been shown that some least-squares based estimates result in SNR-dependent non-vanishing bias of parameter estimates \citep{polzehl2016}. \newtext{Unfortunately, integrating the Rician likelihood into the proposed methodology is non-trivial, as we lose the conjugacy that allows for the analytic form of the posterior \eqref{eqn:function_space_posterior}. To address this issue practically, we suggest a two stage approach which first transforms the noisy Rician/non-central $\chi^2$ signals to noisy Gaussian signals \citep{koay2009noisefloor} and then performs analysis on the latter \citep{chen2019a}.} Additionally, many modern accelerated acquisition schemes induce spatially dependent measurement errors, e.g. from  parallel imaging or sub-sampling and interpolation \citep{griswold2002generalized,fernandez2015}, implying a non-identity covariance matrix in  Equation~\eqref{eqn:data_likelihood}. Incorporating this heteroskedasticity into the proposed method is an important future endeavor.

\section{Conclusion}\label{sec:conclusion}
This work introduces a novel modeling approach and algorithm for fully continuous estimation and uncertainty quantification of the spatially varying ODF field from diffusion MRI data that is robust to sparse angular sampling regimes and low SNRs. We use a latent function-valued random field model and parameterize its random series decomposition using a deep neural network, facilitating flexible data-driven modeling of the complex spatial dependence of the ODF field. Under our proposed model, a closed form conditional posterior predictive distribution is derived and used for facilitating fast ODF estimation and uncertainty quantification. To estimate the network parameters, we employ a penalized maximum likelihood learning scheme which allows us to encode explicit prior assumptions on the shape and the smoothness of the ODFs along with an automatic hyperparameter selection procedure using Bayesian optimization. A post-training calibration procedure is employed to estimate the model's variance parameters. Using extensive simulation studies, we show that our method is able to recover fields with lower error than competing approaches, as well as produce properly calibrated uncertainty quantification across a range of angular sample sizes and SNRs. On a real high resolution diffusion dataset, we demonstrate our method is able to provide superior performance under very sparse and noisy regimes.

\section{Acknowledgements}
The authors would like to acknowledge the following grants which supported this work: R01MH119222, R01MH125860, R01MH132610, R01NS125307.

\section{Software and Supplemental Material}
\newtext{\textbf{Supplement to ``Neural Orientation Distribution Fields for Estimation and Uncertainty Quantification in Diffusion MRI''} includes supporting technical details for the method and algorithm, proofs and derivations of all results, and additional experimental results. The implementation code is publicly accessible at: \textbf{https://github.com/Will-Consagra/NODF}.} 

\bibliographystyle{chicago}
\bibliography{refs}

\begin{thebibliography}{}

\bibitem[\protect\citeauthoryear{Casey, Cannonier, Conley, Cohen, Barch, Heitzeg, Soules, Teslovich, Dellarco, Garavan, Orr, Wager, Banich, Speer, Sutherland, Riedel, Dick, Bjork, Thomas, Chaarani, Mejia, Hagler, {Daniela Cornejo}, Sicat, Harms, Dosenbach, Rosenberg, Earl, Bartsch, Watts, Polimeni, Kuperman, Fair, and Dale}{Casey et~al.}{2018}]{casey2018}
Casey, B., T.~Cannonier, M.~I. Conley, A.~O. Cohen, D.~M. Barch, M.~M. Heitzeg, M.~E. Soules, T.~Teslovich, D.~V. Dellarco, H.~Garavan, C.~A. Orr, T.~D. Wager, M.~T. Banich, N.~K. Speer, M.~T. Sutherland, M.~C. Riedel, A.~S. Dick, J.~M. Bjork, K.~M. Thomas, B.~Chaarani, M.~H. Mejia, D.~J. Hagler, M.~{Daniela Cornejo}, C.~S. Sicat, M.~P. Harms, N.~U. Dosenbach, M.~Rosenberg, E.~Earl, H.~Bartsch, R.~Watts, J.~R. Polimeni, J.~M. Kuperman, D.~A. Fair, and A.~M. Dale (2018).
\newblock The adolescent brain cognitive development (abcd) study: Imaging acquisition across 21 sites.
\newblock {\em Developmental Cognitive Neuroscience\/}~{\em 32}, 43--54.
\newblock The Adolescent Brain Cognitive Development (ABCD) Consortium: Rationale, Aims, and Assessment Strategy.

\bibitem[\protect\citeauthoryear{Frazier}{Frazier}{2018}]{frazier2018}
Frazier, P.~I. (2018).
\newblock A tutorial on bayesian optimization.
\newblock {\em arXiv:stat.ML\/}.

\bibitem[\protect\citeauthoryear{Snoek, Larochelle, and Adams}{Snoek et~al.}{2012}]{snoek2012}
Snoek, J., H.~Larochelle, and R.~P. Adams (2012).
\newblock Practical bayesian optimization of machine learning algorithms.
\newblock In F.~Pereira, C.~Burges, L.~Bottou, and K.~Weinberger (Eds.), {\em Advances in Neural Information Processing Systems}, Volume~25. Curran Associates, Inc.

\bibitem[\protect\citeauthoryear{Solin and S{\"a}rkk{\"a}}{Solin and S{\"a}rkk{\"a}}{2020}]{solin2020}
Solin, A. and S.~S{\"a}rkk{\"a} (2020).
\newblock Hilbert space methods for reduced-rank gaussian process regression.
\newblock {\em Statistics and Computing\/}~{\em 30\/}(2), 419--446.

\bibitem[\protect\citeauthoryear{Zhang, Descoteaux, Zhang, Girard, Chamberland, Dunson, Srivastava, and Zhu}{Zhang et~al.}{2018}]{zhang2018}
Zhang, Z., M.~Descoteaux, J.~Zhang, G.~Girard, M.~Chamberland, D.~Dunson, A.~Srivastava, and H.~Zhu (2018).
\newblock Mapping population-based structural connectomes.
\newblock {\em NeuroImage\/}~{\em 172}, 130--145.

\end{thebibliography}


\begin{thebibliography}{}

\bibitem[\protect\citeauthoryear{Aja-Fernández, Pie¸ciak, and Vegas-Sánchez-Ferrero}{Aja-Fernández et~al.}{2015}]{fernandez2015}
Aja-Fernández, S., T.~Pie¸ciak, and G.~Vegas-Sánchez-Ferrero (2015).
\newblock Spatially variant noise estimation in mri: A homomorphic approach.
\newblock {\em Medical Image Analysis\/}~{\em 20\/}(1), 184--197.

\bibitem[\protect\citeauthoryear{Andersson and Sotiropoulos}{Andersson and Sotiropoulos}{2015}]{andersson2015}
Andersson, J.~L. and S.~N. Sotiropoulos (2015, Nov).
\newblock {{N}on-parametric representation and prediction of single- and multi-shell diffusion-weighted {M}{R}{I} data using {G}aussian processes}.
\newblock {\em Neuroimage\/}~{\em 122}, 166--176.

\bibitem[\protect\citeauthoryear{Basser, Pajevic, Pierpaoli, Duda, and Aldroubi}{Basser et~al.}{2000}]{basser2000}
Basser, P.~J., S.~Pajevic, C.~Pierpaoli, J.~Duda, and A.~Aldroubi (2000).
\newblock In vivo fiber tractography using dt-mri data.
\newblock {\em Magnetic Resonance in Medicine\/}~{\em 44\/}(4), 625--632.

\bibitem[\protect\citeauthoryear{Basser and Pierpaoli}{Basser and Pierpaoli}{1996}]{basser1996}
Basser, P.~J. and C.~Pierpaoli (1996, Jun).
\newblock {{M}icrostructural and physiological features of tissues elucidated by quantitative-diffusion-tensor {M}{R}{I}}.
\newblock {\em J Magn Reson B\/}~{\em 111\/}(3), 209--219.

\bibitem[\protect\citeauthoryear{Becker, Tabelow, Voss, Anwander, Heidemann, and Polzehl}{Becker et~al.}{2012}]{becker2012}
Becker, S.~M., K.~Tabelow, H.~U. Voss, A.~Anwander, R.~M. Heidemann, and J.~Polzehl (2012, Aug).
\newblock {{P}osition-orientation adaptive smoothing of diffusion weighted magnetic resonance data ({P}{O}{A}{S})}.
\newblock {\em Med Image Anal\/}~{\em 16\/}(6), 1142--1155.

\bibitem[\protect\citeauthoryear{Becker, Tabelow, Mohammadi, Weiskopf, and Polzehl}{Becker et~al.}{2014}]{becker2014}
Becker, S. M.~A., K.~Tabelow, S.~Mohammadi, N.~Weiskopf, and J.~Polzehl (2014, July).
\newblock Adaptive smoothing of multi-shell diffusion weighted magnetic resonance data by {msPOAS}.
\newblock {\em Neuroimage\/}~{\em 95}, 90--105.

\bibitem[\protect\citeauthoryear{Berman, Chung, Mukherjee, Hess, Han, and Henry}{Berman et~al.}{2008}]{berman2008}
Berman, J.~I., S.~Chung, P.~Mukherjee, C.~P. Hess, E.~T. Han, and R.~G. Henry (2008, Jan).
\newblock {{P}robabilistic streamline q-ball tractography using the residual bootstrap}.
\newblock {\em Neuroimage\/}~{\em 39\/}(1), 215--222.

\bibitem[\protect\citeauthoryear{Cabeen, Bastin, and Laidlaw}{Cabeen et~al.}{2016}]{cabeen2016}
Cabeen, R.~P., M.~E. Bastin, and D.~H. Laidlaw (2016).
\newblock Kernel regression estimation of fiber orientation mixtures in diffusion mri.
\newblock {\em NeuroImage\/}~{\em 127}, 158--172.

\bibitem[\protect\citeauthoryear{Chen, Dong, Zhang, Lin, Shen, and Yap}{Chen et~al.}{2019}]{chen2019a}
Chen, G., B.~Dong, Y.~Zhang, W.~Lin, D.~Shen, and P.-T. Yap (2019).
\newblock Denoising of diffusion mri data via graph framelet matching in x-q space.
\newblock {\em IEEE Transactions on Medical Imaging\/}~{\em 38\/}(12), 2838--2848.

\bibitem[\protect\citeauthoryear{Chen, Wu, Shen, and Yap}{Chen et~al.}{2019}]{chen2019}
Chen, G., Y.~Wu, D.~Shen, and P.-T. Yap (2019).
\newblock Noise reduction in diffusion mri using non-local self-similar information in joint x\-q space.
\newblock {\em Medical Image Analysis\/}~{\em 53}, 79--94.

\bibitem[\protect\citeauthoryear{Cordero-Grande, Christiaens, Hutter, Price, and Hajnal}{Cordero-Grande et~al.}{2019}]{grande2019}
Cordero-Grande, L., D.~Christiaens, J.~Hutter, A.~N. Price, and J.~V. Hajnal (2019).
\newblock Complex diffusion-weighted image estimation via matrix recovery under general noise models.
\newblock {\em NeuroImage\/}~{\em 200}, 391--404.

\bibitem[\protect\citeauthoryear{de~Micheaux, Mozharovskyi, and Vimond}{de~Micheaux et~al.}{2021}]{micheaux2021}
de~Micheaux, P.~L., P.~Mozharovskyi, and M.~Vimond (2021).
\newblock Depth for curve data and applications.
\newblock {\em Journal of the American Statistical Association\/}~{\em 116\/}(536), 1881--1897.

\bibitem[\protect\citeauthoryear{Descoteaux}{Descoteaux}{2015}]{descoteaux2015}
Descoteaux, M. (2015).
\newblock {\em High Angular Resolution Diffusion Imaging (HARDI)}, pp.\  1--25.
\newblock John Wiley \& Sons, Ltd.

\bibitem[\protect\citeauthoryear{Descoteaux, Angelino, Fitzgibbons, and Deriche}{Descoteaux et~al.}{2007}]{descoteaux2007}
Descoteaux, M., E.~Angelino, S.~Fitzgibbons, and R.~Deriche (2007, Sep).
\newblock {{R}egularized, fast, and robust analytical {Q}-ball imaging}.
\newblock {\em Magn Reson Med\/}~{\em 58\/}(3), 497--510.

\bibitem[\protect\citeauthoryear{Fathony, Sahu, Willmott, and Kolter}{Fathony et~al.}{2021}]{fathony2021}
Fathony, R., A.~K. Sahu, D.~Willmott, and J.~Z. Kolter (2021).
\newblock Multiplicative filter networks.
\newblock In {\em International Conference on Learning Representations}.

\bibitem[\protect\citeauthoryear{G.~de G.~Matthews, Hron, Turner, and Ghahramani}{G.~de G.~Matthews et~al.}{2017}]{matthews2017}
G.~de G.~Matthews, A., J.~Hron, R.~E. Turner, and Z.~Ghahramani (2017).
\newblock Sample-then-optimize posterior sampling for bayesian linear models.
\newblock In {\em Neural Information Processing Systems}.

\bibitem[\protect\citeauthoryear{Galeano, Joseph, and Lillo}{Galeano et~al.}{2015}]{galeano2015}
Galeano, P., E.~Joseph, and R.~E. Lillo (2015).
\newblock The mahalanobis distance for functional data with applications to classification.
\newblock {\em Technometrics\/}~{\em 57\/}(2), 281--291.

\bibitem[\protect\citeauthoryear{Goh, Lenglet, Thompson, and Vidal}{Goh et~al.}{2011}]{goh2011}
Goh, A., C.~Lenglet, P.~M. Thompson, and R.~Vidal (2011, Jun).
\newblock {{A} nonparametric {R}iemannian framework for processing high angular resolution diffusion images and its applications to {O}{D}{F}-based morphometry}.
\newblock {\em Neuroimage\/}~{\em 56\/}(3), 1181--1201.

\bibitem[\protect\citeauthoryear{Golub, Heath, and Wahba}{Golub et~al.}{1979}]{wahba1979}
Golub, G.~H., M.~Heath, and G.~Wahba (1979).
\newblock Generalized cross-validation as a method for choosing a good ridge parameter.
\newblock {\em Technometrics\/}~{\em 21\/}(2), 215--223.

\bibitem[\protect\citeauthoryear{Griswold, Jakob, Heidemann, Nittka, Jellus, Wang, Kiefer, and Haase}{Griswold et~al.}{2002}]{griswold2002generalized}
Griswold, M.~A., P.~M. Jakob, R.~M. Heidemann, M.~Nittka, V.~Jellus, J.~Wang, B.~Kiefer, and A.~Haase (2002).
\newblock Generalized autocalibrating partially parallel acquisitions (grappa).
\newblock {\em Magnetic Resonance in Medicine: An Official Journal of the International Society for Magnetic Resonance in Medicine\/}~{\em 47\/}(6), 1202--1210.

\bibitem[\protect\citeauthoryear{Gudbjartsson and Patz}{Gudbjartsson and Patz}{1995}]{hakon1995}
Gudbjartsson, H. and S.~Patz (1995).
\newblock The rician distribution of noisy mri data.
\newblock {\em Magnetic Resonance in Medicine\/}~{\em 34\/}(6), 910--914.

\bibitem[\protect\citeauthoryear{Guinness and Fuentes}{Guinness and Fuentes}{2016}]{guinness2016}
Guinness, J. and M.~Fuentes (2016).
\newblock Isotropic covariance functions on spheres: Some properties and modeling considerations.
\newblock {\em Journal of Multivariate Analysis\/}~{\em 143}, 143--152.

\bibitem[\protect\citeauthoryear{Haroon, Morris, Embleton, Alexander, and Parker}{Haroon et~al.}{2009}]{haroon2008}
Haroon, H.~A., D.~M. Morris, K.~V. Embleton, D.~C. Alexander, and G.~J.~M. Parker (2009).
\newblock Using the model-based residual bootstrap to quantify uncertainty in fiber orientations from $q$-ball analysis.
\newblock {\em IEEE Transactions on Medical Imaging\/}~{\em 28\/}(4), 535--550.

\bibitem[\protect\citeauthoryear{Henkelman}{Henkelman}{1985}]{henkelman1985}
Henkelman, R.~M. (1985).
\newblock {{M}easurement of signal intensities in the presence of noise in {M}{R} images}.
\newblock {\em Med Phys\/}~{\em 12\/}(2), 232--233.

\bibitem[\protect\citeauthoryear{Hsing and Eubank}{Hsing and Eubank}{2015}]{hsing2015}
Hsing, T. and R.~Eubank (2015).
\newblock {\em Theoretical foundations of functional data analysis, with an introduction to linear operators}.
\newblock John Wiley \& Sons.

\bibitem[\protect\citeauthoryear{Jones, Horsfield, and Simmons}{Jones et~al.}{1999}]{jones1999}
Jones, D., M.~Horsfield, and A.~Simmons (1999).
\newblock Optimal strategies for measuring diffusion in anisotropic systems by magnetic resonance imaging.
\newblock {\em Magnetic Resonance in Medicine\/}~{\em 42\/}(3), 515--525.

\bibitem[\protect\citeauthoryear{Jones}{Jones}{2008}]{jones2008wildboot}
Jones, D.~K. (2008).
\newblock Tractography gone wild: Probabilistic fibre tracking using the wild bootstrap with diffusion tensor mri.
\newblock {\em IEEE Transactions on Medical Imaging\/}~{\em 27\/}(9), 1268--1274.

\bibitem[\protect\citeauthoryear{Kauermann, Claeskens, and Opsomer}{Kauermann et~al.}{2009}]{kauermann2009}
Kauermann, G., G.~Claeskens, and J.~D. Opsomer (2009).
\newblock Bootstrapping for penalized spline regression.
\newblock {\em Journal of Computational and Graphical Statistics\/}~{\em 18\/}(1), 126--146.

\bibitem[\protect\citeauthoryear{Koay, Özarslan, and Basser}{Koay et~al.}{2009}]{koay2009noisefloor}
Koay, C.~G., E.~Özarslan, and P.~J. Basser (2009).
\newblock A signal transformational framework for breaking the noise floor and its applications in mri.
\newblock {\em Journal of Magnetic Resonance\/}~{\em 197\/}(2), 108--119.

\bibitem[\protect\citeauthoryear{Lai and Schumaker}{Lai and Schumaker}{2007}]{lai2007}
Lai, M.-J. and L.~L. Schumaker (2007).
\newblock {\em Spline Functions on Triangulations}.
\newblock Encyclopedia of Mathematics and its Applications. Cambridge University Press.

\bibitem[\protect\citeauthoryear{Liu, Vemuri, and Deriche}{Liu et~al.}{2013}]{liu2013}
Liu, M., B.~C. Vemuri, and R.~Deriche (2013, Feb).
\newblock {{A} robust variational approach for simultaneous smoothing and estimation of {D}{T}{I}}.
\newblock {\em Neuroimage\/}~{\em 67}, 33--41.

\bibitem[\protect\citeauthoryear{Mancini, Jones, and Palombo}{Mancini et~al.}{2022}]{mancini2022}
Mancini, M., D.~K. Jones, and M.~Palombo (2022).
\newblock Lossy compression of multidimensional medical images using sinusoidal activation networks: An evaluation study.
\newblock In {\em Computational Diffusion MRI}, Cham, pp.\  26--37. Springer Nature Switzerland.

\bibitem[\protect\citeauthoryear{Mart{\'i}nez-Hern{\'a}ndez and Genton}{Mart{\'i}nez-Hern{\'a}ndez and Genton}{2020}]{martinezhernadez2020}
Mart{\'i}nez-Hern{\'a}ndez, I. and M.~G. Genton (2020).
\newblock {Recent developments in complex and spatially correlated functional data}.
\newblock {\em Brazilian Journal of Probability and Statistics\/}~{\em 34\/}(2), 204 -- 229.

\bibitem[\protect\citeauthoryear{Mehta, Christinck, Nair, Bussy, Premasiri, Costantino, Chakravarthy, Arnold, Gal, and Arbel}{Mehta et~al.}{2022}]{mehta2022}
Mehta, R., T.~Christinck, T.~Nair, A.~Bussy, S.~Premasiri, M.~Costantino, M.~M. Chakravarthy, D.~L. Arnold, Y.~Gal, and T.~Arbel (2022).
\newblock Propagating uncertainty across cascaded medical imaging tasks for improved deep learning inference.
\newblock {\em IEEE Transactions on Medical Imaging\/}~{\em 41\/}(2), 360--373.

\bibitem[\protect\citeauthoryear{Menafoglio, Secchi, and Rosa}{Menafoglio et~al.}{2013}]{menafoglio2013}
Menafoglio, A., P.~Secchi, and M.~D. Rosa (2013).
\newblock {A Universal Kriging predictor for spatially dependent functional data of a Hilbert Space}.
\newblock {\em Electronic Journal of Statistics\/}~{\em 7\/}(none), 2209 -- 2240.

\bibitem[\protect\citeauthoryear{Michailovich and Rathi}{Michailovich and Rathi}{2010}]{michailovich2010}
Michailovich, O. and Y.~Rathi (2010, Feb).
\newblock {{O}n approximation of orientation distributions by means of spherical ridgelets}.
\newblock {\em IEEE Trans Image Process\/}~{\em 19\/}(2), 461--477.

\bibitem[\protect\citeauthoryear{Michailovich, Rathi, and Dolui}{Michailovich et~al.}{2011}]{michailovich2011}
Michailovich, O., Y.~Rathi, and S.~Dolui (2011).
\newblock Spatially regularized compressed sensing for high angular resolution diffusion imaging.
\newblock {\em IEEE Transactions on Medical Imaging\/}~{\em 30\/}(5), 1100--1115.

\bibitem[\protect\citeauthoryear{Mildenhall, Srinivasan, Tancik, Barron, Ramamoorthi, and Ng}{Mildenhall et~al.}{2021}]{mildenhall2021}
Mildenhall, B., P.~P. Srinivasan, M.~Tancik, J.~T. Barron, R.~Ramamoorthi, and R.~Ng (2021, dec).
\newblock Nerf: Representing scenes as neural radiance fields for view synthesis.
\newblock {\em Commun. ACM\/}~{\em 65\/}(1), 99–106.

\bibitem[\protect\citeauthoryear{Molaei, Aminimehr, Tavakoli, Kazerouni, Azad, Azad, and Merhof}{Molaei et~al.}{2023}]{molaei2023implicit}
Molaei, A., A.~Aminimehr, A.~Tavakoli, A.~Kazerouni, B.~Azad, R.~Azad, and D.~Merhof (2023).
\newblock Implicit neural representation in medical imaging: A comparative survey.
\newblock In {\em Proceedings of the IEEE/CVF International Conference on Computer Vision}, pp.\  2381--2391.

\bibitem[\protect\citeauthoryear{M\"uller, Evans, Schied, and Keller}{M\"uller et~al.}{2022}]{mueller2022}
M\"uller, T., A.~Evans, C.~Schied, and A.~Keller (2022, July).
\newblock Instant neural graphics primitives with a multiresolution hash encoding.
\newblock {\em ACM Trans. Graph.\/}~{\em 41\/}(4), 102:1--102:15.

\bibitem[\protect\citeauthoryear{Ning, Laun, Gur, DiBella, Deslauriers-Gauthier, Megherbi, Ghosh, Zucchelli, Menegaz, Fick, St-Jean, Paquette, Aranda, Descoteaux, Deriche, O'Donnell, and Rathi}{Ning et~al.}{2015}]{ning2015}
Ning, L., F.~Laun, Y.~Gur, E.~V. DiBella, S.~Deslauriers-Gauthier, T.~Megherbi, A.~Ghosh, M.~Zucchelli, G.~Menegaz, R.~Fick, S.~St-Jean, M.~Paquette, R.~Aranda, M.~Descoteaux, R.~Deriche, L.~O'Donnell, and Y.~Rathi (2015).
\newblock {{S}parse {R}econstruction {C}hallenge for diffusion {M}{R}{I}: {V}alidation on a physical phantom to determine which acquisition scheme and analysis method to use?}
\newblock {\em Medical Image Analysis\/}~{\em 26\/}(1), 316--331.

\bibitem[\protect\citeauthoryear{Novikov, Fieremans, Jespersen, and Kiselev}{Novikov et~al.}{2019}]{novikov2019quantifying}
Novikov, D.~S., E.~Fieremans, S.~N. Jespersen, and V.~G. Kiselev (2019).
\newblock Quantifying brain microstructure with diffusion mri: Theory and parameter estimation.
\newblock {\em NMR in Biomedicine\/}~{\em 32\/}(4), e3998.

\bibitem[\protect\citeauthoryear{Novikov, Veraart, Jelescu, and Fieremans}{Novikov et~al.}{2018}]{novikov2018}
Novikov, D.~S., J.~Veraart, I.~O. Jelescu, and E.~Fieremans (2018).
\newblock Rotationally-invariant mapping of scalar and orientational metrics of neuronal microstructure with diffusion mri.
\newblock {\em NeuroImage\/}~{\em 174}, 518--538.

\bibitem[\protect\citeauthoryear{Ombao, Lindquist, Thompson, and Aston}{Ombao et~al.}{2016}]{ombao2016}
Ombao, H., M.~Lindquist, W.~Thompson, and J.~Aston (Eds.) (2016).
\newblock {\em Handbook of Neuroimaging Data Analysis\/} (1st ed.).
\newblock Chapman and Hall/CRC.

\bibitem[\protect\citeauthoryear{Polzehl and Tabelow}{Polzehl and Tabelow}{2016}]{polzehl2016}
Polzehl, J. and K.~Tabelow (2016).
\newblock Low snr in diffusion mri models.
\newblock {\em Journal of the American Statistical Association\/}~{\em 111\/}(516), 1480--1490.

\bibitem[\protect\citeauthoryear{Quellmalz}{Quellmalz}{2020}]{quellmalz2020}
Quellmalz, M. (2020, Aug).
\newblock The funk--radon transform for hyperplane sections through a common point.
\newblock {\em Analysis and Mathematical Physics\/}~{\em 10\/}(3), 38.

\bibitem[\protect\citeauthoryear{Rahaman, Baratin, Arpit, Draxler, Lin, Hamprecht, Bengio, and Courville}{Rahaman et~al.}{2019}]{rahaman2019}
Rahaman, N., A.~Baratin, D.~Arpit, F.~Draxler, M.~Lin, F.~Hamprecht, Y.~Bengio, and A.~Courville (2019, 09--15 Jun).
\newblock On the spectral bias of neural networks.
\newblock In K.~Chaudhuri and R.~Salakhutdinov (Eds.), {\em Proceedings of the 36th International Conference on Machine Learning}, Volume~97 of {\em Proceedings of Machine Learning Research}, pp.\  5301--5310. PMLR.

\bibitem[\protect\citeauthoryear{Raj, Hess, and Mukherjee}{Raj et~al.}{2011}]{raj2011}
Raj, A., C.~Hess, and P.~Mukherjee (2011).
\newblock Spatial hardi: Improved visualization of complex white matter architecture with bayesian spatial regularization.
\newblock {\em NeuroImage\/}~{\em 54\/}(1), 396--409.

\bibitem[\protect\citeauthoryear{Ramos-Llordén, nchez Ferrero, Liao, Westin, Setsompop, and Rathi}{Ramos-Llordén et~al.}{2021}]{ramos2021}
Ramos-Llordén, G., G.~nchez Ferrero, C.~Liao, C.~F. Westin, K.~Setsompop, and Y.~Rathi (2021, Sep).
\newblock {{S}{N}{R}-enhanced diffusion {M}{R}{I} with structure-preserving low-rank denoising in reproducing kernel {H}ilbert spaces}.
\newblock {\em Magn Reson Med\/}~{\em 86\/}(3), 1614--1632.

\bibitem[\protect\citeauthoryear{Rasmussen and Williams}{Rasmussen and Williams}{2005}]{rasmussen2005}
Rasmussen, C.~E. and C.~K.~I. Williams (2005).
\newblock {\em Gaussian Processes for Machine Learning (Adaptive Computation and Machine Learning)}.
\newblock The MIT Press.

\bibitem[\protect\citeauthoryear{Saragadam, LeJeune, Tan, Balakrishnan, Veeraraghavan, and Baraniuk}{Saragadam et~al.}{2023}]{saragadam2023}
Saragadam, V., D.~LeJeune, J.~Tan, G.~Balakrishnan, A.~Veeraraghavan, and R.~G. Baraniuk (2023).
\newblock Wire: Wavelet implicit neural representations.
\newblock {\em arXiv: cs.CV\/}.

\bibitem[\protect\citeauthoryear{Schilling, Fadnavis, Batson, Visagie, Combes, By, McKnight, Bagnato, Garyfallidis, Landman, Smith, and O'Grady}{Schilling et~al.}{2023}]{schilling2023}
Schilling, K.~G., S.~Fadnavis, J.~Batson, M.~Visagie, A.~J. Combes, S.~By, C.~D. McKnight, F.~Bagnato, E.~Garyfallidis, B.~A. Landman, S.~A. Smith, and K.~P. O'Grady (2023).
\newblock Denoising of diffusion mri in the cervical spinal cord – effects of denoising strategy and acquisition on intra-cord contrast, signal modeling, and feature conspicuity.
\newblock {\em NeuroImage\/}~{\em 266}, 119826.

\bibitem[\protect\citeauthoryear{Schwartzman, Dougherty, and Taylor}{Schwartzman et~al.}{2008}]{schwartzman2008}
Schwartzman, A., R.~F. Dougherty, and J.~E. Taylor (2008).
\newblock False discovery rate analysis of brain diffusion direction maps.
\newblock {\em The Annals of Applied Statistics\/}~{\em 2\/}(1), 153--175.

\bibitem[\protect\citeauthoryear{Siddiqui, Höllt, and Vilanova}{Siddiqui et~al.}{2021}]{siddiqui2021}
Siddiqui, F., T.~Höllt, and A.~Vilanova (2021).
\newblock A progressive approach for uncertainty visualization in diffusion tensor imaging.
\newblock {\em Computer Graphics Forum\/}~{\em 40\/}(3), 411--422.

\bibitem[\protect\citeauthoryear{Sitzmann, Martel, Bergman, Lindell, and Wetzstein}{Sitzmann et~al.}{2020}]{sitzmann2020}
Sitzmann, V., J.~N.~P. Martel, A.~W. Bergman, D.~B. Lindell, and G.~Wetzstein (2020).
\newblock Implicit neural representations with periodic activation functions.
\newblock In {\em Proceedings of the 34th International Conference on Neural Information Processing Systems}, NIPS'20, Red Hook, NY, USA. Curran Associates Inc.

\bibitem[\protect\citeauthoryear{Sjölund, Eklund, Özarslan, Herberthson, Bånkestad, and Knutsson}{Sjölund et~al.}{2018}]{jens2018}
Sjölund, J., A.~Eklund, E.~Özarslan, M.~Herberthson, M.~Bånkestad, and H.~Knutsson (2018).
\newblock Bayesian uncertainty quantification in linear models for diffusion mri.
\newblock {\em NeuroImage\/}~{\em 175}, 272--285.

\bibitem[\protect\citeauthoryear{Snoek, Larochelle, and Adams}{Snoek et~al.}{2012}]{snoek2012}
Snoek, J., H.~Larochelle, and R.~P. Adams (2012).
\newblock Practical bayesian optimization of machine learning algorithms.
\newblock In F.~Pereira, C.~Burges, L.~Bottou, and K.~Weinberger (Eds.), {\em Advances in Neural Information Processing Systems}, Volume~25. Curran Associates, Inc.

\bibitem[\protect\citeauthoryear{Snoek, Rippel, Swersky, Kiros, Satish, Sundaram, Patwary, Prabhat, and Adams}{Snoek et~al.}{2015}]{snoek2015}
Snoek, J., O.~Rippel, K.~Swersky, R.~Kiros, N.~Satish, N.~Sundaram, M.~M.~A. Patwary, P.~Prabhat, and R.~P. Adams (2015).
\newblock Scalable bayesian optimization using deep neural networks.
\newblock In {\em Proceedings of the 32nd International Conference on International Conference on Machine Learning - Volume 37}, ICML'15, pp.\  2171–2180. JMLR.org.

\bibitem[\protect\citeauthoryear{Sun, Xie, Ye, Ho, Entezari, Blackband, and Vemuri}{Sun et~al.}{2013}]{baba2013}
Sun, J., Y.~Xie, W.~Ye, J.~Ho, A.~Entezari, S.~J. Blackband, and B.~C. Vemuri (2013).
\newblock Dictionary learning on the manifold of square root densities and application to reconstruction of diffusion propagator fields.
\newblock In J.~C. Gee, S.~Joshi, K.~M. Pohl, W.~M. Wells, and L.~Z{\"o}llei (Eds.), {\em Information Processing in Medical Imaging}, Berlin, Heidelberg, pp.\  619--631. Springer Berlin Heidelberg.

\bibitem[\protect\citeauthoryear{Tancik, Casser, Yan, Pradhan, Mildenhall, Srinivasan, Barron, and Kretzschmar}{Tancik et~al.}{2022}]{tancik2022blocknerf}
Tancik, M., V.~Casser, X.~Yan, S.~Pradhan, B.~Mildenhall, P.~Srinivasan, J.~T. Barron, and H.~Kretzschmar (2022).
\newblock {Block-NeRF}: Scalable large scene neural view synthesis.
\newblock {\em arXiv: cs.CV\/}.

\bibitem[\protect\citeauthoryear{Tancik, Srinivasan, Mildenhall, Fridovich-Keil, Raghavan, Singhal, Ramamoorthi, Barron, and Ng}{Tancik et~al.}{2020}]{tancik2020}
Tancik, M., P.~Srinivasan, B.~Mildenhall, S.~Fridovich-Keil, N.~Raghavan, U.~Singhal, R.~Ramamoorthi, J.~Barron, and R.~Ng (2020).
\newblock Fourier features let networks learn high frequency functions in low dimensional domains.
\newblock In H.~Larochelle, M.~Ranzato, R.~Hadsell, M.~Balcan, and H.~Lin (Eds.), {\em Advances in Neural Information Processing Systems}, Volume~33, pp.\  7537--7547. Curran Associates, Inc.

\bibitem[\protect\citeauthoryear{Tanno, Worrall, Kaden, Ghosh, Grussu, Bizzi, Sotiropoulos, Criminisi, and Alexander}{Tanno et~al.}{2021}]{tanno2021}
Tanno, R., D.~E. Worrall, E.~Kaden, A.~Ghosh, F.~Grussu, A.~Bizzi, S.~N. Sotiropoulos, A.~Criminisi, and D.~C. Alexander (2021).
\newblock Uncertainty modelling in deep learning for safer neuroimage enhancement: Demonstration in diffusion mri.
\newblock {\em NeuroImage\/}~{\em 225}, 117366.

\bibitem[\protect\citeauthoryear{Tournier, Calamante, and Connelly}{Tournier et~al.}{2012}]{tournier2012mrtrix}
Tournier, J.-D., F.~Calamante, and A.~Connelly (2012).
\newblock Mrtrix: diffusion tractography in crossing fiber regions.
\newblock {\em International journal of imaging systems and technology\/}~{\em 22\/}(1), 53--66.

\bibitem[\protect\citeauthoryear{Tournier, Yeh, Calamante, Cho, Connelly, and Lin}{Tournier et~al.}{2008}]{tournier2008resolving}
Tournier, J.-D., C.-H. Yeh, F.~Calamante, K.-H. Cho, A.~Connelly, and C.-P. Lin (2008).
\newblock Resolving crossing fibres using constrained spherical deconvolution: validation using diffusion-weighted imaging phantom data.
\newblock {\em Neuroimage\/}~{\em 42\/}(2), 617--625.

\bibitem[\protect\citeauthoryear{Tuch}{Tuch}{2004}]{tuch2004}
Tuch, D.~S. (2004).
\newblock {Q‐ball imaging}.
\newblock {\em Magnetic Resonance in Medicine\/}~{\em 52}, 1358--1372.

\bibitem[\protect\citeauthoryear{Veraart, Novikov, Christiaens, Ades-aron, Sijbers, and Fieremans}{Veraart et~al.}{2016}]{veraart2016}
Veraart, J., D.~S. Novikov, D.~Christiaens, B.~Ades-aron, J.~Sijbers, and E.~Fieremans (2016).
\newblock Denoising of diffusion mri using random matrix theory.
\newblock {\em NeuroImage\/}~{\em 142}, 394--406.

\bibitem[\protect\citeauthoryear{Veraart, Nunes, Rudrapatna, Fieremans, Jones, Novikov, and Shemesh}{Veraart et~al.}{2020}]{veraart2020noninvasive}
Veraart, J., D.~Nunes, U.~Rudrapatna, E.~Fieremans, D.~K. Jones, D.~S. Novikov, and N.~Shemesh (2020).
\newblock Noninvasive quantification of axon radii using diffusion mri.
\newblock {\em Elife\/}~{\em 9}, e49855.

\bibitem[\protect\citeauthoryear{Wang, Dong, Tian, Liao, Fan, Hoge, Keil, Polimeni, Wald, Huang, and Setsompop}{Wang et~al.}{2021}]{wang2021}
Wang, F., Z.~Dong, Q.~Tian, C.~Liao, Q.~Fan, W.~S. Hoge, B.~Keil, J.~R. Polimeni, L.~L. Wald, S.~Y. Huang, and K.~Setsompop (2021, Apr).
\newblock In vivo human whole-brain connectom diffusion mri dataset at 760{\thinspace}{\textmu}m isotropic resolution.
\newblock {\em Scientific Data\/}~{\em 8\/}(1), 122.

\bibitem[\protect\citeauthoryear{Xie, Takikawa, Saito, Litany, Yan, Khan, Tombari, Tompkin, sitzmann, and Sridhar}{Xie et~al.}{2022}]{xie2022}
Xie, Y., T.~Takikawa, S.~Saito, O.~Litany, S.~Yan, N.~Khan, F.~Tombari, J.~Tompkin, V.~sitzmann, and S.~Sridhar (2022).
\newblock Neural fields in visual computing and beyond.
\newblock {\em Computer Graphics Forum\/}~{\em 41\/}(2), 641--676.

\bibitem[\protect\citeauthoryear{Yap, An, Chen, and Shen}{Yap et~al.}{2014}]{yap2014}
Yap, P.-T., H.~An, Y.~Chen, and D.~Shen (2014).
\newblock Uncertainty estimation in diffusion mri using the nonlocal bootstrap.
\newblock {\em IEEE Transactions on Medical Imaging\/}~{\em 33\/}(8), 1627--1640.

\bibitem[\protect\citeauthoryear{Ye, Zhuo, Gullapalli, and Prince}{Ye et~al.}{2016}]{ye2016}
Ye, C., J.~Zhuo, R.~P. Gullapalli, and J.~L. Prince (2016).
\newblock Estimation of fiber orientations using neighborhood information.
\newblock {\em Medical Image Analysis\/}~{\em 32}, 243--256.

\bibitem[\protect\citeauthoryear{Yu and Li}{Yu and Li}{2013}]{yu2013}
Yu, T. and P.~Li (2013).
\newblock Spatial shrinkage estimation of diffusion tensors on diffusion-weighted imaging data.
\newblock {\em Journal of the American Statistical Association\/}~{\em 108\/}(503), 864--875.

\bibitem[\protect\citeauthoryear{Yuce, Ortiz-Jimenez, Besbinar, and Frossard}{Yuce et~al.}{2022}]{gizem2022}
Yuce, G., G.~Ortiz-Jimenez, B.~Besbinar, and P.~Frossard (2022).
\newblock A structured dictionary perspective on implicit neural representations.
\newblock {\em arXiv:cs.LG\/}.

\bibitem[\protect\citeauthoryear{Zhang, Daducci, He, Schiavi, Seguin, Smith, Yeh, Zhao, and O’Donnell}{Zhang et~al.}{2022}]{fanzhang2022}
Zhang, F., A.~Daducci, Y.~He, S.~Schiavi, C.~Seguin, R.~E. Smith, C.-H. Yeh, T.~Zhao, and L.~J. O’Donnell (2022).
\newblock Quantitative mapping of the brain’s structural connectivity using diffusion mri tractography: A review.
\newblock {\em NeuroImage\/}~{\em 249}, 118870.

\end{thebibliography}

\clearpage
\pagebreak

\begin{center}
{\large\bf Supplement to ``Neural Orientation Distribution Fields for Estimation and Uncertainty Quantification in Diffusion MRI''}
\end{center}

\setcounter{equation}{0}
\setcounter{figure}{0}
\setcounter{table}{0}
\setcounter{section}{0}
\setcounter{page}{1}
\makeatletter
\renewcommand{\theequation}{S\arabic{equation}}
\renewcommand{\thefigure}{S\arabic{figure}}
\renewcommand{\thetable}{S\arabic{table}}
\renewcommand{\thesection}{S\arabic{section}}
\renewcommand{\bibnumfmt}[1]{[S#1]}
\renewcommand{\citenumfont}[1]{S#1}
\renewcommand{\theequation}{S.\arabic{equation}}
\renewcommand{\thesection}{S\arabic{section}}
\renewcommand{\thesubsection}{S\arabic{section}.\arabic{subsection}}
\renewcommand{\thetable}{S\arabic{table}}
\renewcommand{\thefigure}{S\arabic{figure}}
\renewcommand{\thetheorem}{S\arabic{theorem}}
\renewcommand{\theproposition}{S\arabic{proposition}}
\renewcommand{\thelemma}{S\arabic{lemma}}
\renewcommand{\theassumption}{S\arabic{assumption}}

\section{Bayesian Optimization for Hyperparamter Tuning}\label{apx:BO_hyperparam_tuning}
Denote $x=(\lambda_c, \text{rest})\in \mathcal{X}$, where $\text{rest}$ refers to optional additional hyper-parameters, e.g. $\gamma$, $r$, learning rate, number of epochs, etc, to be selected and $\mathcal{X}$ is some predefined parameter range. Denote the partition of the observed data into disjoint sets 
$$
\{\boldsymbol{Y}, \boldsymbol{V}\} = \{\boldsymbol{Y}_{test}, \boldsymbol{V}_{test}\}\bigcup\{\boldsymbol{Y}_{train}, \boldsymbol{V}_{train}\}.
$$
Our hyperparameter optimization routine aims to identify the optimal hyperparameter configuration $x$ which, when used to maximize the penalized likelihood \eqref{eqn:penalized_likelihood} using $\{\boldsymbol{Y}_{train}, \boldsymbol{V}_{train}\}$, produces $(\widehat{\boldsymbol{\mu}}, \widehat{\boldsymbol{W}}, \widehat{\boldsymbol{\theta}})$ which maximize 
the data likelihood \eqref{eqn:data_likelihood} of 
$\{\boldsymbol{Y}_{test}, \boldsymbol{V}_{test}\}$, conditioned on $(\widehat{\boldsymbol{\mu}}, \widehat{\boldsymbol{W}}, \widehat{\boldsymbol{\theta}})$. We accomplish this using the framework of Bayesian optimization (BO) \citeSupp{snoek2012,frazier2018}, in which a surrogate Gaussian process is used to model the stochastic map $x \mapsto \mathcal{L}(x)$, where $\mathcal{L}$ denotes the data likelihood \eqref{eqn:data_likelihood} evaluated on the test set conditioned on parameters learned from setting $x$, along with an acquisition function which allows for exploration/exploitation trade-off. The map is evaluated using Algorithm~\ref{alg:eval_prediction_error}, where the stochasticity results from the data randomization in step 2 and the random initialization in the optimization step 3 of the procedure.
\begin{algorithm}[!ht]
  \caption{Algorithm to evaluate expensive stochastic function $\mathcal{L}$}
  \label{alg:eval_prediction_error}
  \begin{algorithmic}[1]
    \State \textbf{Input} Data $(\boldsymbol{Y},\boldsymbol{V})$, hyper-parameter setting $x$, train-test proportion split $p$
    \State Randomly partition data into disjoint sets 
    $\{\boldsymbol{Y}, \boldsymbol{V}\} = \{\boldsymbol{Y}_{test}, \boldsymbol{V}_{test}\}\bigcup\{\boldsymbol{Y}_{train}, \boldsymbol{V}_{train}\}$ according to split proportion $p$
    \State Form estimates 
         $(\widehat{\boldsymbol{\mu}}, \widehat{\boldsymbol{W}}, \widehat{\boldsymbol{\theta}})$ by maximizing \eqref{eqn:penalized_likelihood} using $\{\boldsymbol{Y}_{train}, \boldsymbol{V}_{train}\}$ under given hyper-parameter configuration $x$
         using stochastic gradient descent
    \State Return likelihood evaluation \eqref{eqn:data_likelihood} on $\{\boldsymbol{Y}_{test}, \boldsymbol{V}_{test}\}$ conditioned on $(\widehat{\boldsymbol{\mu}}, \widehat{\boldsymbol{W}}, \widehat{\boldsymbol{\theta}})$
\end{algorithmic}
\end{algorithm}
\par 
In our implementation, the surrogate GP prior modeling $\mathcal{L}(x)$ is specified with a constant mean function and a ARD Mat\'{e}rn 5/2 covariance kernel. Such design choices encode a-priori notations of an invariant mean, e.g. we don't know a priori which parts of the hyperparameter space will exhibit better performance, and the fact that closer pairs of points in the hyperparameter space should be more highly correlated. A simple and popular choice for the acquisition function is to sequentially maximize the candidates expected improvement, defined as  
\begin{equation}\label{eqn:EI_acquisition}
    \text{EI}_n(x) := \mathbb{E}[[\mathcal{L}(x) - \mathcal{L}_{n}(x)]^{+} | x_{1}, \mathcal{L}(x_{1}),...,x_{n}, \mathcal{L}(x_{n})] 
\end{equation}
where $\mathcal{L}_n(x):= \max_{m\le n} \mathcal{L}(x_{m})$ and $[a]^{+}=\text{max}(0,a)$ and the expectation is taken with respect to the GP posterior. This function can be evaluated in closed form under the GP surrogate model. The acquisition function then maximizes $\text{EI}_n$ to select the $n+1$'st candidate hyperparameter setting. Intuitively, this procedure selects the $x$ which we expect the largest difference between the yet unobserved value of $\mathcal{L}(x)$ and the current maximizer $\mathcal{L}(x)$. 
At the onset, an initial low-discrepancy sequence is generated in the parameter space and is used to estimate the hyperparameters of the GP emulators kernel.

\begin{algorithm}[!ht]
  \caption{BO Algorithm for hyper-parameter optimization}
  \label{alg:BO_hyperparameter_optimization}
  \begin{algorithmic}[1]
    \State \textbf{Input} Data $(\boldsymbol{Y},\boldsymbol{V})$, hyper-parameter ranges $\mathcal{X}$, max number of trials
    \State Generate $n_0$ points in hyper-parameter space chosen via initial low-discrepancy sequence
    \State Sample $(\mathcal{L}((x_1), ..., \mathcal{L}(x_{n_{0}}))$ 
    using Algorithm~\ref{alg:eval_prediction_error}
    \While {$n \le \text{max number of trials}$}
        \State \text{Set $x_{n}$ to be the maximizer of \eqref{eqn:EI_acquisition}}
        \State \text{Sample $\mathcal{L}(x_n)$ using Algorithm~\ref{alg:eval_prediction_error}}
        \State \text{$n \gets n+1$}
    \EndWhile
    \State \text{Return hyper-parameter configurion $x_m$ corresponding to the smallest $\mathcal{L}(x_m)$}
\end{algorithmic}
\end{algorithm}
\section{Proofs and Derivations}\label{apx:theory}
\textbf{Proof of Proposition \ref{prop:eigen_analysis_S2}}:
\begin{proof}
By Bochner's theorem, the stationarity of $C_{\mathcal{H}}$ implies we can compute the process spectral density function as 
$$
    \text{s}_{\gamma}(\omega) = \int \text{exp}(-i\omega t)h_{\mathbb{S}^2}(t)dt.
$$
Furthermore, it is well established that rotationally invariant spherical kernels commute with $\Delta_{\mathbb{S}^2}$, and hence share the same eigenfunctions, namely the spherical harmonics. Then
using the results of \citeSupp{solin2020}, we have the decomposition:
$$
Cor(\boldsymbol{p}_1,\boldsymbol{p}_2) = \sum_{k=0}^\infty s_{\gamma}(\sqrt{l_k(l_{k}+1)})\phi_k(\boldsymbol{p}_1)\phi_k(\boldsymbol{p}_2),
$$
where the odd harmonics are annihilated from the fact that the realizations of the process must be antipodally symmetric w.p.1. 
\end{proof}
\noindent{\textbf{Derivation of Equation~\eqref{eqn:function_space_posterior}}}
\par\bigskip
We begin by showing that, given the parametric representation \eqref{eqn:lin_field_model_params}, the prior assumptions laid out in Section~\ref{ssec:latent_density_model} can be translated to the following conditional weight-space prior on $\boldsymbol{W}$:
\begin{proposition}\label{prop:weight_space_prior}
Under~\eqref{eqn:lin_field_model_params}, the prior distribution 
\begin{equation}\label{eqn:weight_matrix_prior}
\boldsymbol{W}|\boldsymbol{\theta},\gamma, \sigma_w^2 \overset{\text{dist}}{=} \mathcal{MN}_{K\times r}(\boldsymbol{0},\boldsymbol{R}_{\gamma}^{-1}, \sigma_w^2\boldsymbol{I}_{r})
\end{equation}
implies a Gaussian process prior on $g_{\boldsymbol{v}}$ with reduced rank correlation function \eqref{eqn:zonal_mercer_kernel}, where $\sigma_w^2 > 0$ is a prior variance parameter.
\end{proposition}
\begin{proof}
Given the model \eqref{eqn:lin_field_model_params}, it follows that 
$$
\mathbb{E}[g(\boldsymbol{v},\boldsymbol{p})] = \boldsymbol{\mu}^{\intercal}\boldsymbol{\xi}_{\boldsymbol{\theta}}(\boldsymbol{v}),
$$
and 
$$
\begin{aligned}
    \text{Cov}(g(\boldsymbol{v},\boldsymbol{p}_1),g(\boldsymbol{v},\boldsymbol{p}_2)) &= \mathbb{E}[\boldsymbol{\phi}^{\intercal}(\boldsymbol{p}_1)\boldsymbol{W}\boldsymbol{\xi}_{\boldsymbol{\theta}}(\boldsymbol{v})][\boldsymbol{\xi}^{\intercal}_{\boldsymbol{\theta}}(\boldsymbol{v})\boldsymbol{W}^{\intercal}]\boldsymbol{\phi}(\boldsymbol{p}_2)]\\
    &= \boldsymbol{\phi}^{\intercal}(\boldsymbol{p}_1)\mathbb{E}[\boldsymbol{W}\boldsymbol{\xi}_{\boldsymbol{\theta}}(\boldsymbol{v})\boldsymbol{\xi}^{\intercal}_{\boldsymbol{\theta}}(\boldsymbol{v})\boldsymbol{W}^{\intercal}]\boldsymbol{\phi}(\boldsymbol{p}_2)\\
    &=  \boldsymbol{\phi}^{\intercal}(\boldsymbol{p}_1)[\boldsymbol{R}^{-1}_{\gamma}\text{trace}(\sigma_w^2\boldsymbol{I}_r\boldsymbol{\xi}_{\boldsymbol{\theta}}(\boldsymbol{v})\boldsymbol{\xi}^{\intercal}_{\boldsymbol{\theta}}(\boldsymbol{v}))]\boldsymbol{\phi}(\boldsymbol{p}_2)\\
    &= \sigma_w^2\boldsymbol{\xi}^{\intercal}_{\boldsymbol{\theta}}(\boldsymbol{v})\boldsymbol{\xi}_{\boldsymbol{\theta}}(\boldsymbol{v})\boldsymbol{\phi}^{\intercal}(\boldsymbol{p}_1)\boldsymbol{R}_{\gamma}^{-1}\boldsymbol{\phi}^{\intercal}(\boldsymbol{p}_2), \\
\end{aligned}
$$
where the expectation are taken w.r.t. to the prior \eqref{eqn:weight_matrix_prior}. Hence the correlation function of $g_{\boldsymbol{v}}$ is $ \boldsymbol{\phi}^{\intercal}(\boldsymbol{p}_1)\boldsymbol{R}_{\gamma}^{-1}\boldsymbol{\phi}^{\intercal}(\boldsymbol{p}_2)$, a rank $K$ approximation to \eqref{eqn:zonal_mercer_kernel}, as desired to show. 
\end{proof}
Denote the centered data
$\boldsymbol{Y}^{(c)}:=\boldsymbol{Y}-\boldsymbol{1}_{M}\boldsymbol{\mu}^{\intercal}\boldsymbol{\Xi}_{\boldsymbol{\theta}}$. Using properties of the matrix normal distribution and the well known identity: 
$
\text{vec}(\boldsymbol{\Phi}_{G}\boldsymbol{W}\boldsymbol{\Xi}_{\boldsymbol{\theta}}) = [\boldsymbol{\Xi}_{\boldsymbol{\theta}}^{\intercal}\otimes\boldsymbol{\Phi}_{G}]\text{vec}(\boldsymbol{W}),
$ 
where $\text{vec}$ is the vectorization operator and $\otimes$ is the Kronecker product, the likelihood \eqref{eqn:data_likelihood} and prior \eqref{eqn:weight_matrix_prior} can be vectorized into the following two-stage hierarchical model:
\begin{equation}\label{eqn:two_stage_HBM}
\begin{aligned}
    &\text{vec}(\boldsymbol{Y}^{(c)})|\boldsymbol{V}, \text{vec}(\boldsymbol{W}), \boldsymbol{\theta}, \boldsymbol{\mu}, \sigma_e^2  \sim \mathcal{N}_{Mn}([\boldsymbol{\Xi}_{\boldsymbol{\theta}}^{\intercal}\otimes\boldsymbol{\Phi}_{G}]\text{vec}(\boldsymbol{W}), \sigma_e^2\boldsymbol{I}_{MN}) \\
    &\text{vec}(\boldsymbol{W})|\boldsymbol{\theta}, \sigma^2_{w}, \gamma \sim \mathcal{N}_{Kr}(\boldsymbol{0}, \sigma^2_{w}\boldsymbol{I}_{r}\otimes\boldsymbol{R_{\gamma}}^{-1}).
\end{aligned}
\end{equation}
Using Gaussian-Gaussian conjugacy, we have the closed form conditional posterior:
\begin{equation}\label{eqn:weight_posterior}
\text{vec}(\boldsymbol{W}) | \boldsymbol{V}, \boldsymbol{Y}, \boldsymbol{\theta}, \boldsymbol{\mu}, \gamma, \sigma_w^2,  \sigma_e^2 \sim \mathcal{N}_{Kr}(\frac{1}{\sigma_{e}^2}\boldsymbol{\Lambda}_{\boldsymbol{\theta}}^{-1}[\boldsymbol{\Xi}_{\boldsymbol{\theta}}^{\intercal}\otimes\boldsymbol{\Phi}_{G}]^{\intercal}\text{vec}(\boldsymbol{Y}^{(c)}), \boldsymbol{\Lambda}_{\boldsymbol{\theta}}^{-1})
\end{equation}
where 
\begin{equation}\label{eqn:basis_covariance_matrix}
\begin{aligned}
    \boldsymbol{\Lambda}_{\boldsymbol{\theta}} = \frac{1}{\sigma^2_{e}}(\frac{\sigma^2_{e}}{\sigma^2_{w}}\boldsymbol{I}_{r}\otimes\boldsymbol{R} +  \boldsymbol{\Xi}_{\boldsymbol{\theta}}\boldsymbol{\Xi}_{\boldsymbol{\theta}}^{\intercal}\otimes\boldsymbol{\Phi}_{G}^{\intercal}\boldsymbol{\Phi}_{G}).
    \end{aligned}
\end{equation}
Due to the closure of the normal distribution under linear transformation, the posterior over the weights~\eqref{eqn:weight_posterior} induces the following multivariate normal posterior over the coefficient field 
\begin{equation}\label{eqn:coef_field_posterior}
\begin{aligned}
    \boldsymbol{c}(\boldsymbol{v}) | \boldsymbol{V}, \boldsymbol{Y}, \boldsymbol{\theta}, \boldsymbol{\mu}, \gamma, \sigma_w^2,  \sigma_e^2 \sim \mathcal{N}_{K}(&\frac{1}{\sigma_{e}^2}[\boldsymbol{\xi}_{\boldsymbol{\theta}}^{\intercal}(\boldsymbol{v})\otimes\boldsymbol{I}_{K}]\boldsymbol{\Lambda}_{\boldsymbol{\theta}}^{-1}[\boldsymbol{\Xi}_{\boldsymbol{\theta}}^{\intercal}\otimes\boldsymbol{\Phi}_{G}]^{\intercal}\text{vec}(\boldsymbol{Y}^{(c)}), \\
    &[\boldsymbol{\xi}_{\boldsymbol{\theta}}^{\intercal}(\boldsymbol{v})\otimes\boldsymbol{I}_{K}]\boldsymbol{\Lambda}_{\boldsymbol{\theta}}^{-1}[\boldsymbol{\xi}_{\boldsymbol{\theta}}^{\intercal}(\boldsymbol{v})\otimes\boldsymbol{I}_{K}]^{\intercal}).
    \end{aligned}
\end{equation}
Recall that we assume the conditional predictive distribution of the mean function is:
$\mu(\boldsymbol{v}) |\boldsymbol{\theta}, \boldsymbol{\mu}, \sigma_{\mu}^2 \sim \mathcal{N}(\boldsymbol{\mu}^{\intercal}\boldsymbol{\xi}_{\boldsymbol{\theta}}(\boldsymbol{v}), \sigma_\mu^2)$,
and independence: $\mu(\boldsymbol{v})\perp \boldsymbol{c}(\boldsymbol{v}) | \boldsymbol{\theta}$. Coupling this with \eqref{eqn:coef_field_posterior} and using the properties of independent Gaussians, it is straight forward to derive the joint distribution:  
$$
\begin{aligned}
    \begin{pmatrix}
    \mu(\boldsymbol{v}) \\
    \boldsymbol{c}(\boldsymbol{v}))
    \end{pmatrix} | \boldsymbol{V}, \boldsymbol{Y}, \boldsymbol{\theta}, \boldsymbol{\mu}, \gamma, \sigma_w^2,  \sigma_e^2, \sigma_\mu^2 \sim \mathcal{N}\big(\begin{bmatrix}
        \boldsymbol{\mu}^{\intercal}\boldsymbol{\xi}_{\boldsymbol{\theta}}(\boldsymbol{v}) \\
        \frac{1}{\sigma_{e}^2}[\boldsymbol{\xi}_{\boldsymbol{\theta}}^{\intercal}(\boldsymbol{v})\otimes\boldsymbol{I}_{K}]\boldsymbol{\Lambda}_{\boldsymbol{\theta}}^{-1}[\boldsymbol{\Xi}_{\boldsymbol{\theta}}^{\intercal}\otimes\boldsymbol{\Phi}_{G}]^{\intercal}\text{vec}(\boldsymbol{Y}^{(c)}),
    \end{bmatrix} \\
    \begin{bmatrix}
                    \sigma_u^2 & \boldsymbol{0} \\
                    \boldsymbol{0} & [\boldsymbol{\xi}_{\boldsymbol{\theta}}^{\intercal}(\boldsymbol{v})\otimes\boldsymbol{I}_{K}]\boldsymbol{\Lambda}_{\boldsymbol{\theta}}^{-1}[\boldsymbol{\xi}_{\boldsymbol{\theta}}^{\intercal}(\boldsymbol{v})\otimes\boldsymbol{I}_{K}]^{\intercal}
    \end{bmatrix} \big).
    \end{aligned}
$$
Given that 
$$
g(\boldsymbol{v},\boldsymbol{p}) = \mu(\boldsymbol{v}) + \boldsymbol{c}^{\intercal}(\boldsymbol{v})\boldsymbol{\phi}(\boldsymbol{p}) =  \begin{pmatrix} 1 & \boldsymbol{\phi}(\boldsymbol{p}) \end{pmatrix}    \begin{pmatrix}
    \mu(\boldsymbol{v}) \\
    \boldsymbol{c}(\boldsymbol{v}))
    \end{pmatrix},
$$
using standard properties of the covariance, it follows that 
$$
\begin{aligned}
    g(\boldsymbol{v},\cdot)| \boldsymbol{V}, \boldsymbol{Y}, \boldsymbol{\theta}, \boldsymbol{\mu}, \gamma, \sigma_w^2,  \sigma_e^2, \sigma_\mu^2 \sim \mathcal{GP}\big(&\boldsymbol{\xi}_{\boldsymbol{\theta}}^{\intercal}(\boldsymbol{v})\boldsymbol{\mu} + \boldsymbol{\phi}^{\intercal}(\boldsymbol{p})\frac{1}{\sigma_{e}^2}[\boldsymbol{\xi}_{\boldsymbol{\theta}}^{\intercal}(\boldsymbol{v})\otimes\boldsymbol{I}_{K}]\boldsymbol{\Lambda}_{\boldsymbol{\theta}}^{-1}[\boldsymbol{\Xi}_{\boldsymbol{\theta}}^{\intercal}\otimes\boldsymbol{\Phi}_{G}]^{\intercal}\text{vec}(\boldsymbol{Y}^{(c)}), \\ &\sigma_{\mu}^2 + \boldsymbol{\phi}^{\intercal}(\boldsymbol{p}_1)[\boldsymbol{\xi}_{\boldsymbol{\theta}}^{\intercal}(\boldsymbol{v})\otimes\boldsymbol{I}_{K}]\boldsymbol{\Lambda}_{\boldsymbol{\theta}}^{-1}[\boldsymbol{\xi}_{\boldsymbol{\theta}}^{\intercal}(\boldsymbol{v})\otimes\boldsymbol{I}_{K}]^{\intercal}\boldsymbol{\phi}(\boldsymbol{p}_2)\big),
\end{aligned}
$$
which is \eqref{eqn:function_space_posterior}, as desired. 
\par\bigskip
\noindent{\textbf{Invertibility of the Funk-Radon Transform}}
\begin{proposition}\label{prop:FR_invertibility}
    The Funk-Radon transform, denoted $\mathcal{G}^{-1}$, is invertible on $\text{span}\left(\{\phi_0, ..., \phi_K\}\right)\subset\mathcal{H}$, for any $K$.
\end{proposition}
\begin{proof}
    For any $h\in\mathcal{H}$, we have the expansion
    $$
        h = \sum_{k=0}^\infty a_k\phi_k \approx \sum_{k=0}^K a_k\phi_k 
    $$
    for some $l2$-convergent sequence of coefficient $\{a_k\}$. Using the Funk-Hecke theorem, it can be shown that the real-symmetric spherical harmonics $\{\phi_k\}$ are non-zero eigenfunctions of $\mathcal{G}^{-1}$, with corresponding eigenvalue $2\pi P_{l_{k}}(0)$, where $P_{l_{k}}(0)$ is the Legendre polynomial of degree $l_{k}$, corresponding to the degree of $\phi_k$. Coupling this with the linearity of $\mathcal{G}^{-1}$, it is easy to show that 
    $$
    \mathcal{G}^{-1}(h) = \mathcal{G}^{-1}\left(\sum_{k=0}^\infty a_k\phi_k\right) = \sum_{k=0}^\infty 2\pi P_{l_{k}}(0) a_k\phi_k \approx \sum_{k=0}^K 2\pi P_{l_{k}}(0) a_k\phi_k
    $$
    Then $\mathcal{G}$ restricted to $\text{span}\left(\{\phi_0, ..., \phi_K\}\right)$ is the linear transform with spectrum $\{\phi_k,[2\pi P_{l_{k}}(0)]^{-1}\}$, which is unique. 
\end{proof}
The above analysis can be extended to infinite $K$ by assuming some stronger smoothness class than $L^2(\mathbb{S})$, but is beyond the scope of interest here as we always assume finite rank $K$ approximation.

\section{Additional Implementation Details and Experimental Results}\label{apx:additional_results}

\subsection{Brief Overview of the Bootstrap}\label{ssec:bootstrap_discussion}

In this section, we provide details on the residual bootstrap procedure. Specifically, a pilot estimator is first fit using a ridge regression on the raw or local-PCA-based smoothed signals via \eqref{eqn:ridge_regression}. The residual $\hat{\boldsymbol{\epsilon}}_i = \boldsymbol{y}_i - \boldsymbol{\Phi}\hat{\boldsymbol{c}}_{\boldsymbol{v}_{i}}$ under model~\eqref{eqn:stat_model_observed_data} has second order moment $\mathbb{E}[\hat{\boldsymbol{\epsilon}}_i^2] = \sigma_e^2d_i$, where $d_i = [(\boldsymbol{I} - \boldsymbol{H}_{\lambda})(\boldsymbol{I} - \boldsymbol{H}_{\lambda})]_{ii}$, and $\boldsymbol{H}_{\lambda}$ is the smoothing matrix of the ridge regression~\eqref{eqn:ridge_regression}. Therefore, we implement the residual adjustment $\tilde{\boldsymbol{\epsilon}}_i = \hat{\boldsymbol{\epsilon}}_i/\sqrt{d_i}$, which helps mitigate the small sample bias of the bootstrap. Bootstrapped samples $\boldsymbol{\epsilon}^{*}_i$ are drawn with replacement from  $\tilde{\boldsymbol{\epsilon}}_i$ and used to construct bootstrapped signals $\boldsymbol{y}^{*}_i$, and ultimately bootstrapped coefficient estimates $\boldsymbol{c}^{*}_{\boldsymbol{v}_{i}}$ via applying SHLS-Raw or  SHLS-MPPCA to $\boldsymbol{y}^{*}_i$, which are then used for inference. 
\par 
We clarify that for SHLS-MPPCA, the local PCA-based low-rank model was integrated into the bootstrapping procedure. Specifically, the algorithm is given by:
\begin{enumerate}
    \item Apply the MPPCA to the raw data and form a pilot estimator via \eqref{eqn:ridge_regression} on the smoothed data.
    \item Generate residuals from the raw data by subtracting the pilot estimator from step 1.
    \item Resample with replacement from the rescaled  the residuals.
    \item Apply MPPCA on each bootstrapped sample and refit \eqref{eqn:ridge_regression} on the smoothed data.
\end{enumerate}
We notice that for bootstrapping MPPCA type denoising approaches, it is typical in the literature to apply  a repetition bootstrap, i.e. acquiring multiple measurements per $\boldsymbol{p}_m$ and resampling among them for all gradient directions independently \citep{veraart2016,schilling2023}. However, this of course requires repeated measures, which are often not collected and are infeasible in the super sparse acquisition regime due to the increase in scanning time, and hence this approach was not considered here.
\subsection{A Note on Computation}\label{ssec:computation_bootstrap}
We now discuss some computational details of the inference procedures. All experiments were run on a Linux machine equipped with a NVIDIA Titan
RTX GPU with 24GB of RAM. The dipy implementation of MPPCA took on average $\approx$ 23 seconds to run on the 2D ROI considered in Section~\ref{ssec:real_data}. The cost of inference by bootstrapping the whole procedure scales linearly in $B$, with $B$ being the number of bootstrapped samples. Generating a reasonable number of samples for inference, say $B=500$ as is used in our analysis, takes $\approx 3.2$ hours. It is often the case that we are interested in quantifying the uncertainty in a very computationally expensive downstream procedure, e.g. tractography, or for many subjects, rendering such time demands undesirable in practice. As has been noted previously \citep{jens2018}, this underscores the issues with bootstrap-based uncertainty quantification in large scale medical imaging analysis for all but the fastest of estimators, e.g. SHLS~\eqref{eqn:ridge_regression} which has a closed from analytic solution and thus can be computed rapidly, due to the scale of the data. Finally, we note that as MPPCA is relatively fast compared to alternative denoising approaches, the issue of using bootstrap-based uncertainty quantification would only be further exacerbated for more computationally intensive denoising approaches, e.g. deep learning based, if they do not allow for native uncertainty quantification. 
\par 
For NODF, the major computational expenditure in Algorithm~\ref{alg:inferece_algo} is in retraining the neural field multiple times in the hyperparameter optimization scheme in step 2. For example, for the largest neural field with $r=256$ and $L=3$ trained on the 2D ROI, training took $\approx$ 21 seconds. In our analysis, we set the total number of trials in Algorithm~\ref{alg:BO_hyperparameter_optimization} to be 20, hence obtaining the final estimate $\widehat{\boldsymbol{\theta}}$ took $\approx 7$ minutes. Once this step has been completed, the remainder of the procedure is computationally trivial. The calibration set in step 4 can be formed with a small number of voxels ($\approx 50\text{ to }100$) and the 2D grid of candidate optimizing values 
can also be made small, e.g. we found success using a marginal grid of 5 values for each variance parameter, so solving \eqref{eqn:calibration_likelihood_optimization} is fast. Forming the inverse \eqref{eqn:basis_covariance_matrix} is fast ($<0.5$ seconds for $K=45$, $r=256$) and only needs to be done once and is then stored. As the posterior is defined continuously, the remaining cost of inference is forming the mean and covariance function defining \eqref{eqn:function_space_posterior}, which only require (relatively small) matrix multiplications and are therefore fast. In contrast, proper bootstrap-based uncertainty estimates ``off the grid'' would require incorporating whatever interpolation scheme was employed in the resampled estimates, further exacerbating the computational issues. 
\par 
We conclude with two remarks on the practical implementation of our procedure. If predictions need to be formed for an exceedingly large collection of voxels, these computations are fully parallelizable over this set and thus can be run concurrently. Additionally, if it is desirable to analyze the brain images for many subjects from a single acquisition protocol, e.g. the same $M$, $b$-value, $\sigma_e^2$, etc., the hyperparameter selection can be run once and then fixed for the remaining subjects. 

\subsection{Additional Results from Synthetic and Real Data Experiments}
\begin{table}[!ht]
    \centering
    \scriptsize
\begin{tabular}{llrrrrrr}
\toprule
    &    & \multicolumn{3}{l}{M=10} & \multicolumn{3}{l}{M=60} \\
    &    &    NODF & SHLS-Raw & SHLS-MPPCA &    NODF & SHLS-Raw & SHLS-MPPCA \\
 &  &         &          &            &         &          &            \\
\midrule
ODF & ECP & 5.00e-04 & 1.20e-03 &   2.90e-03 & 1.00e-03 & 3.20e-03 &   3.00e-04 \\
    & IL & 3.00e-04 & 1.00e-04 &   4.00e-04 & 3.00e-04 & 3.40e-03 &   2.00e-04 \\
    & $L^2$-Error & 3.00e-04 & 4.00e-04 &   3.00e-04 & 3.00e-04 & 2.00e-04 &   2.00e-04 \\
GFA & ECP & 1.70e-03 & 4.60e-03 &   2.61e-02 & 3.10e-03 & 5.80e-03 &   3.10e-03 \\
    & IL & 3.00e-04 & 1.00e-04 &   5.00e-04 & 5.00e-04 & 6.20e-03 &   1.00e-04 \\
    & Abs. Error & 4.00e-04 & 4.00e-04 &   1.40e-03 & 3.00e-04 & 1.10e-03 &   2.00e-04 \\
    & Bias & 4.00e-04 & 4.00e-04 &   1.50e-03 & 2.00e-04 & 2.00e-03 &   2.00e-04 \\
\bottomrule
\end{tabular}

    \caption{Standard errors for the Monte-Carlo simulation averages reported in Table~\ref{tab:2d_crossing_analysis_qualitative}.}
\label{tab:sim_2d_results_standard_errors}
\end{table}

\begin{table}[!ht]
    \centering
    \addtolength{\leftskip} {-2cm}
    \addtolength{\rightskip}{-2cm}
    \scriptsize
\begin{tabular}{llrrrrrrrrr}
\toprule
   &    & \multicolumn{3}{l}{NODF} & \multicolumn{3}{l}{SHLS-Raw} & \multicolumn{3}{l}{SHLS-MPPCA} \\
   &    & $L^2$ & ECP & IL & $L^2$ & ECP & IL  & $L^2$ & ECP & IL \\
\midrule
10 & 20 &     1.52e-03 &       1.06e-04 &     3.14e-04 &     1.25e-04 &       2.15e-04 &     7.70e-06 &     1.64e-04 &        5.64e-04 &     2.64e-05 \\
   & 10 &     9.19e-04 &       1.17e-04 &     1.04e-03 &     2.79e-04 &       3.79e-04 &     6.23e-05 &     3.79e-04 &      6.44e-04 &     5.39e-05 \\
20 & 20 &     5.91e-04 &       1.64e-04 &     6.37e-04 &     1.06e-04 &       1.42e-04 &     3.13e-05 &     1.17e-04 &       5.07e-04 &     2.18e-05 \\
   & 10 &     6.40e-04 &       5.35e-04 &     1.30e-03 &     2.25e-04 &       1.46e-04 &     6.25e-05 &     2.00e-04 &        5.26e-04 &     5.06e-05 \\
30 & 20 &     6.41e-04 &       5.17e-04 &     6.90e-04 &     8.45e-05 &       7.09e-05 &     2.59e-05 &     9.02e-05 &       4.06e-04 &     2.33e-05 \\
   & 10 &     5.65e-04 &       9.38e-04 &     6.71e-04 &     1.83e-04 &       6.62e-05 &     5.23e-05 &     1.54e-04 &      4.51e-04 &     5.47e-05 \\
40 & 20 &     3.37e-04 &       6.63e-04 &     2.38e-04 &     6.27e-05 &       4.32e-05 &     1.84e-05 &     8.46e-05 &      4.53e-04 &     2.31e-05 \\
   & 10 &     4.98e-04 &       1.10e-03 &     2.17e-05 &     1.32e-04 &       3.49e-05 &     3.67e-05 &     1.71e-04 &     4.85e-04 &     5.12e-05 \\
50 & 20 &     4.40e-04 &       1.07e-03 &     1.20e-03 &     5.92e-05 &       3.69e-05 &     1.92e-05 &     7.93e-05 &       1.30e-03 &     4.18e-04 \\
   & 10 &     4.37e-04 &       1.34e-03 &     8.47e-04 &     1.24e-04 &       1.57e-05 &     3.90e-05 &     1.29e-04 &     2.90e-04 &     1.01e-04 \\
60 & 20 &     2.51e-04 &       7.20e-04 &     1.99e-04 &     5.10e-05 &       2.62e-05 &     1.86e-05 &     5.79e-05 &       8.82e-04 &     3.54e-04 \\
   & 10 &     4.04e-04 &       1.21e-03 &     8.50e-04 &     1.10e-04 &       8.72e-06 &     3.58e-05 &     1.18e-04 &       2.08e-04 &     8.91e-05 \\
\bottomrule
\end{tabular}
\caption{Standard errors for the Monte-Carlo simulation averages reported in Table~\ref{tab:sim_3d_pointwise_results}.}
\label{tab:sim_3d_results_standard_errors}
\end{table}
Tables~\ref{tab:sim_2d_results_standard_errors} and \ref{tab:sim_3d_results_standard_errors} provide the standard errors for the MC averages reported in Tables~\ref{tab:2d_crossing_analysis_qualitative} and Table~\ref{tab:sim_3d_pointwise_results}, respectively. 

\subsection{ABCD Data}\label{ssec:abcd_subject}
\begin{figure}[!ht]
    \centering
    \includegraphics[scale=0.62]{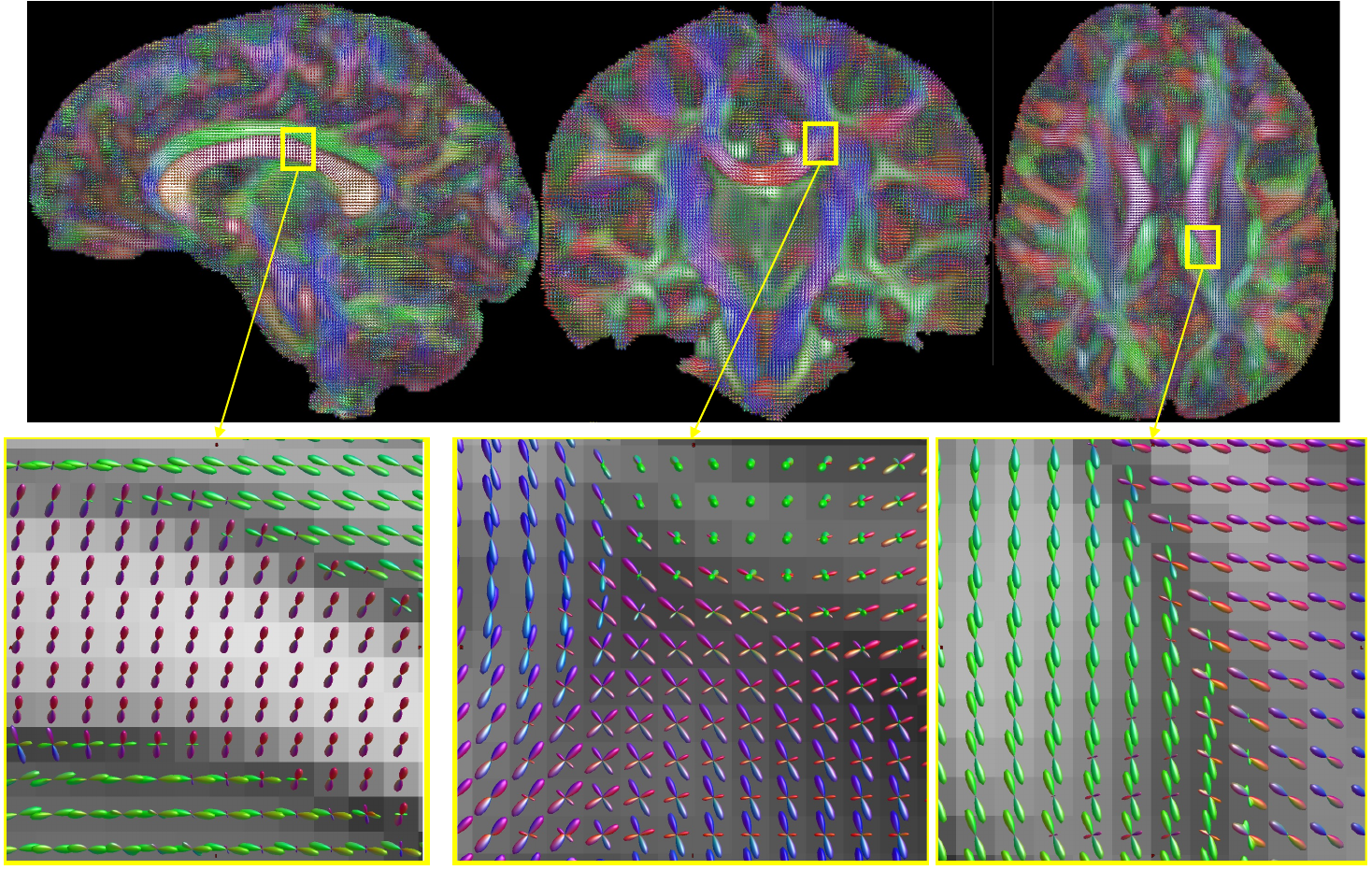}
    \caption{\newtext{Sagital (left), coronal (middle) and axial (right) views of full-brain NODF fits for a randomly selected ABCD subject. ODFs were sharpened with CSD to enhance directional information. Background is colored by estimated GFA.}}
    \label{fig:abcd_results}
\end{figure}
NODF's strong performance in the low SNR, sparse sample regime considered in Section~\ref{sec:experiments} is due to the deep basis functions ability to learn spatial correlations in the field and share this information to improve statistical efficiency in prediction. It is also important to validate that our method can perform reliable estimation of the field in the relatively high SNR and dense sampling regimes encountered in many applications. \newtext{Moreover, the fields considered in Section~\ref{sec:experiments} are relatively small subsets of the overall image, and hence we would also like to evaluate how NODF scales to full-brain fits.}
\par 
To investigate this, we apply our method to diffusion data from the Adolescent Brain Cognitive Development (ABCD) study \citeSupp{casey2018}. The ABCD dataset is a large-scale, longitudinal study that aims to track brain development and other health outcomes in over 10,000 children and adolescents in the United States. Full imaging and acquisition protocol for the diffusion scans can be found \citeSupp{casey2018}. \newtext{We chose the ABCD dataset for this evaluation because it includes resolutions and signal-to-noise ratios that are more typical in diffusion MRI studies. We processed the diffusion data of a randomly selected test subject using the publicly available PSC pipeline \citeSupp{zhang2018}. Preprocessing steps include eddy correction, motion correction, and $b=0$ inhomogeneity correction and brain masking. We use the b=3,000 ($M=60$) shell data for analysis.} 
\par 
\newtext{NODF was fit using the approach outlined in Section~\ref{sec:stat_inference} and the parameters discussed in Section~\ref{ssec:implementation_details} for $b=3,000$ data, except we set $L=8$ and $r=1,024$ in order to accommodate the increased signal complexity of the full-brain image. We also find that for numerical stability and convergence in optimization, a reduced learning rate of $10^{-6}$ was preferable. We trained the field for $2,500$ iterations.} 
\par 
\newtext{The top row of Figure~\ref{fig:abcd_results} shows the resulting full-brain field estimates, with spherical deconvolution applied to the estimated ODFs to improve visual clarity of the directional information. We find there is a strong concurrence between the inferred directions obtained from the field and our expectations based on anatomical information across the white matter regions.  
For instance, in the zoomed in ROI in the bottom middle panel, we see the characteristic crossing fiber region expected in the centrum semiovale.}

\bibliographystyleSupp{chicago}
\bibliographySupp{refs}

\end{document}